\documentclass[a4paper]{article}

\title{From Algebraic to coordinate Bethe ansatz for square ice}

\author{Silvere Gangloff (LIP, ENS Lyon)}

\usepackage{latexsym,amsthm,amscd}
\usepackage{amssymb}
\usepackage{amsfonts}
\usepackage{amsmath}
\usepackage{physics}
\usepackage[all]{xy}
\usepackage{rotating}
\usepackage[utf8]{inputenc}
\usepackage[T1]{fontenc}
\usepackage[english]{babel}
\usepackage[svgnames]{xcolor}
\usepackage{geometry}
\usepackage{lastpage}
\usepackage{array}
\usepackage{stmaryrd}
\usepackage{hyperref}
\usepackage{tikz}
\usetikzlibrary{patterns}
\usepgflibrary{decorations.pathmorphing}
\usetikzlibrary{decorations.pathmorphing}
\usetikzlibrary{calc,shapes.arrows,decorations.pathreplacing}
\usetikzlibrary{shadows}
\usetikzlibrary{positioning}

\newtheorem{proposition}{Proposition}
\newtheorem{notation}{Notation}
\newtheorem{definition}{Definition}
\newtheorem{theorem}{Theorem}
\newtheorem{lemma}{Lemma}
\newtheorem{computation}{Computation}
\newtheorem{fact}{Fact}
\newtheorem{remark}{Remark}
\newtheorem{property}{Property}
\newtheorem{example}{Example}

\renewcommand{\vec}[1]{\textbf{#1}}

\usetikzlibrary{arrows}

\begin{document}
\maketitle

\begin{abstract}

In this text, we provide a detailed 
exposition of the Algebraic Bethe ansatz for 
square ice (or six vertex model), which 
allows the construction of candidate 
eigenvectors for the transfer matrices of 
this model. We also 
prove some formula of 
V.E. Korepin for these vectors, which leads
to an identification, up to a non-zero 
complex factor, with the 
vector obtained by coordinate Bethe ansatz.

\end{abstract}

\section{Introduction}

The model of ice was introduced 
as a model of statistical mechanics by L. Pauling~\cite{Pauling} in
a three-dimensional setting, and consists 
in the set of possible arrangements of 
dihydrogen monoxide molecules on an infinite grid.
He proposed a first very simple 
approximation of entropy for this model. 
This approximation was refined by further works [see for instance~\cite{Stillinger}], leading 
to an exact computation by E.H.Lieb~\cite{Lieb67}, 
using mainly H. Bethe method~\cite{Bethe} for the diagonalisation 
of the Hamiltonian of the XXZ model, the analytic 
work of C.N. Yang and C.P. Yang~\cite{YY66}~\cite{YY66II}, 
and the diagonalisation of the XY model 
by E.H. Lieb, T. Shultz and D. Mattis~\cite{LSM61}. This work led to an extensive 
study of what are now called exactly solved 
models. Notably, the work of R.Baxter 
led to the discovery of a rich algebraic structure 
underlying these models, and an algebraisation 
of Bethe method, called now the algebraic Bethe ansatz, that he applied for the eight-vertex model~\cite{baxter}, an extension of 
the six-vertex model - itself a representation 
of square ice.
However, the computation done by E.H.Lieb was not 
rigorous, in particular since it
relied on some unproven hypothesis, and 
some arguments left partial in the texts 
on which it relied.
Motivated by percolation theory, rigorous proofs 
of some parts of C.N. Yang and C.P. Yang statements were recently done by H. Duminil-Copin, M. Gagnebin, M. Harel, I. Manolescu and V. Tassion~\cite{Duminil-Copin}. They also proposed 
an exposition of Bethe method, called coordinate 
Bethe ansatz~\cite{Duminil-Copin.ansatz}. 
We then provided 
a complete proof, which involves an exposition 
of E.H.Lieb's arguments and some techniques 
developped later, as well as completion of several partial arguments, that 
square ice entropy is equal to 
$\frac{3}{2} \log_2 (4/3)$.
However, our aim was to extend the methods 
developped in this field to other models,  
called multidimensional subshifts of finite type, 
that lie outside of the scope of 
statistical physics, since they appear 
at the frontier between mathematics and computation theories. The proof 
that we presented makes use of 
the coordinate Bethe ansatz,
which is highly specific to square ice, 
as it relies on some structures of $\mathbb{Z}^2$ 
that appear naturally when one represents 
this models as a discrete curves model. 
On the other hand, the algebraic Bethe ansatz 
can be formulated for the broad class 
of nearest-neighbour multidimensional subshifts 
of finite type, as soon as one can find 
a solution of the so-called Yang-Baxter equation 
for this model. We expect that this method, applied to particular subshifts of finite type, 
would lead to new exactly solved models.
In the present text, we propose 
a complete exposition of the algebraic Bethe ansatz for square ice. Since the proof of square ice entropy 
relies on an expression of the candidate eigenvector 
provided by the coordinate Bethe ansatz, we prove 
a formula of V.E. Korepin which leads to an identification, 
up to a non-zero complex factor, to the vectors 
obtained by both ansatz. We also provide a proof 
of the commutation of the transfer matrices of the model with 
some Heisenberg Hamiltonian which relies on development of 
the algebraic Bethe 
ansatz.
In the perspective 
of a generalisation to other models, this derivations 
would help us deriving all necessary properties 
of the vector only from the algebraic Bethe ansatz. 
Although these 
results are not completely new, since the literature on 
the subject is dispersed and does not provide 
complete proofs, we expect that 
this exposition will be useful for a broad 
audience. \bigskip

The text is organized as follows: in Section~\ref{section.background}, we present square ice and its 
representations [Section~\ref{section.square.ice.representations}],
define Lieb transfer matrices [Section~\ref{section.lieb.path}], and state 
the coordinate Bethe ansatz [Section~\ref{section.coordinate.bethe.ansatz}]. 
In Section~\ref{section.overview}, 
we present an overview of the remaining 
of the text.

\section{\label{section.background} Background}

\subsection{\label{section.square.ice.representations} Representations of square ice}

The most widely used representation of 
square ice is the six vertex model 
(whose name derives from that the elements 
of the alphabet represent vertices of a regular 
grid) and is presented in Section~\ref{section.six.vertex}. In this text, we 
will use another representation, presented in 
Section~\ref{section.discrete.curves}, 
whose configurations consist of drifting 
discrete curves, representing possible 
particle trajectories.

\subsubsection{\label{section.six.vertex} The six vertex model}

The \textbf{six vertex model} is the the 
set of elements of $\mathcal{A}_0^{\mathbb{Z}^2}$, 
where 

\[\mathcal{A}_0 = \left\{ \begin{tikzpicture}[scale=0.6,baseline=2mm]
\draw (0,0) rectangle (1,1);
\draw[-latex,line width=0.2mm] (0.5,0) -- (0.5,0.5);
\draw[-latex,line width=0.2mm] (0.5,0.5) -- (0.5,1); 
\draw[-latex,line width=0.2mm] (0,0.5) -- (0.5,0.5);
\draw[-latex,line width=0.2mm] (0.5,0.5) -- (1,0.5); 
\end{tikzpicture}, \begin{tikzpicture}[scale=0.6,baseline=2mm]
\draw (0,0) rectangle (1,1);
\draw[latex-,line width=0.2mm] (0.5,0) -- (0.5,0.5);
\draw[latex-,line width=0.2mm] (0.5,0.5) -- (0.5,1); 
\draw[-latex,line width=0.2mm] (0,0.5) -- (0.5,0.5);
\draw[-latex,line width=0.2mm] (0.5,0.5) -- (1,0.5); 
\end{tikzpicture}, \begin{tikzpicture}[scale=0.6,baseline=2mm]
\draw (0,0) rectangle (1,1);
\draw[latex-,line width=0.2mm] (0.5,0) -- (0.5,0.5);
\draw[-latex,line width=0.2mm] (0.5,0.5) -- (0.5,1); 
\draw[-latex,line width=0.2mm] (0,0.5) -- (0.5,0.5);
\draw[latex-,line width=0.2mm] (0.5,0.5) -- (1,0.5); 
\end{tikzpicture}, \begin{tikzpicture}[scale=0.6,baseline=2mm]
\draw (0,0) rectangle (1,1);
\draw[-latex,line width=0.2mm] (0.5,0) -- (0.5,0.5);
\draw[latex-,line width=0.2mm] (0.5,0.5) -- (0.5,1); 
\draw[latex-,line width=0.2mm] (0,0.5) -- (0.5,0.5);
\draw[-latex,line width=0.2mm] (0.5,0.5) -- (1,0.5); 
\end{tikzpicture}, \begin{tikzpicture}[scale=0.6,baseline=2mm]
\draw (0,0) rectangle (1,1);
\draw[latex-,line width=0.2mm] (0.5,0) -- (0.5,0.5);
\draw[latex-,line width=0.2mm] (0.5,0.5) -- (0.5,1); 
\draw[latex-,line width=0.2mm] (0,0.5) -- (0.5,0.5);
\draw[latex-,line width=0.2mm] (0.5,0.5) -- (1,0.5); 
\end{tikzpicture}, 
\begin{tikzpicture}[scale=0.6,baseline=2mm]
\draw (0,0) rectangle (1,1);
\draw[-latex,line width=0.2mm] (0.5,0) -- (0.5,0.5);
\draw[-latex,line width=0.2mm] (0.5,0.5) -- (0.5,1); 
\draw[latex-,line width=0.2mm] (0,0.5) -- (0.5,0.5);
\draw[latex-,line width=0.2mm] (0.5,0.5) -- (1,0.5); 
\end{tikzpicture}\right\},\]
such that for any two adjacent positions in 
$\mathbb{Z}^2$, the arrows corresponding 
to the common edge of the symbols on the 
two positions have to be directed the same way: 
for instance, the pattern $\begin{tikzpicture}[scale=0.6,baseline=2mm]
\begin{scope}
\draw (0,0) rectangle (1,1);
\draw[-latex,line width=0.2mm] (0.5,0) -- (0.5,0.5);
\draw[-latex,line width=0.2mm] (0.5,0.5) -- (0.5,1); 
\draw[-latex,line width=0.2mm] (0,0.5) -- (0.5,0.5);
\draw[-latex,line width=0.2mm] (0.5,0.5) -- (1,0.5); 
\end{scope} 

\begin{scope}[xshift=1cm]
\draw (0,0) rectangle (1,1);
\draw[latex-,line width=0.2mm] (0.5,0) -- (0.5,0.5);
\draw[latex-,line width=0.2mm] (0.5,0.5) -- (0.5,1); 
\draw[-latex,line width=0.2mm] (0,0.5) -- (0.5,0.5);
\draw[-latex,line width=0.2mm] (0.5,0.5) -- (1,0.5); 
\end{scope}
\end{tikzpicture}$ is allowed, 
while $\begin{tikzpicture}[scale=0.6,baseline=2mm]
\begin{scope}
\draw (0,0) rectangle (1,1);
\draw[-latex,line width=0.2mm] (0.5,0) -- (0.5,0.5);
\draw[-latex,line width=0.2mm] (0.5,0.5) -- (0.5,1); 
\draw[-latex,line width=0.2mm] (0,0.5) -- (0.5,0.5);
\draw[-latex,line width=0.2mm] (0.5,0.5) -- (1,0.5); 
\end{scope} 

\begin{scope}[xshift=1cm]
\draw (0,0) rectangle (1,1);
\draw[-latex,line width=0.2mm] (0.5,0) -- (0.5,0.5);
\draw[latex-,line width=0.2mm] (0.5,0.5) -- (0.5,1); 
\draw[latex-,line width=0.2mm] (0,0.5) -- (0.5,0.5);
\draw[-latex,line width=0.2mm] (0.5,0.5) -- (1,0.5); 
\end{scope}
\end{tikzpicture}$ is not. \bigskip

In an element of the model, 
the symbols draw a grid whose 
edges are oriented in such a way 
that all the vertices have two incoming 
arrows and two outgoing ones. This 
is called an Eulerian orientation 
of the square lattice. See an illustration
on Figure~\ref{figure.six.vertex}.

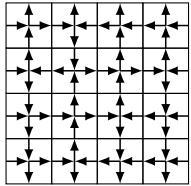
\begin{figure}[h!]
\[\begin{tikzpicture}[scale=0.6]

\begin{scope}
\draw (0,0) rectangle (1,1);
\draw[latex-,line width=0.2mm] (0.5,0) -- (0.5,0.5);
\draw[latex-,line width=0.2mm] (0.5,0.5) -- (0.5,1); 
\draw[-latex,line width=0.2mm] (0,0.5) -- (0.5,0.5);
\draw[-latex,line width=0.2mm] (0.5,0.5) -- (1,0.5); 
\end{scope} 

\begin{scope}[xshift=1cm]
\draw (0,0) rectangle (1,1);
\draw[latex-,line width=0.2mm] (0.5,0) -- (0.5,0.5);
\draw[-latex,line width=0.2mm] (0.5,0.5) -- (0.5,1); 
\draw[-latex,line width=0.2mm] (0,0.5) -- (0.5,0.5);
\draw[latex-,line width=0.2mm] (0.5,0.5) -- (1,0.5); 
\end{scope}

\begin{scope}[xshift=2cm]
\draw (0,0) rectangle (1,1);
\draw[latex-,line width=0.2mm] (0.5,0) -- (0.5,0.5);
\draw[latex-,line width=0.2mm] (0.5,0.5) -- (0.5,1); 
\draw[latex-,line width=0.2mm] (0,0.5) -- (0.5,0.5);
\draw[latex-,line width=0.2mm] (0.5,0.5) -- (1,0.5); 
\end{scope}

\begin{scope}[xshift=3cm]
\draw (0,0) rectangle (1,1);
\draw[latex-,line width=0.2mm] (0.5,0) -- (0.5,0.5);
\draw[latex-,line width=0.2mm] (0.5,0.5) -- (0.5,1); 
\draw[latex-,line width=0.2mm] (0,0.5) -- (0.5,0.5);
\draw[latex-,line width=0.2mm] (0.5,0.5) -- (1,0.5); 
\end{scope}

\begin{scope}[yshift=1cm]
\draw (0,0) rectangle (1,1);
\draw[latex-,line width=0.2mm] (0.5,0) -- (0.5,0.5);
\draw[latex-,line width=0.2mm] (0.5,0.5) -- (0.5,1); 
\draw[-latex,line width=0.2mm] (0,0.5) -- (0.5,0.5);
\draw[-latex,line width=0.2mm] (0.5,0.5) -- (1,0.5); 
\end{scope} 

\begin{scope}[xshift=1cm,yshift=1cm]
\draw (0,0) rectangle (1,1);
\draw[-latex,line width=0.2mm] (0.5,0) -- (0.5,0.5);
\draw[-latex,line width=0.2mm] (0.5,0.5) -- (0.5,1); 
\draw[-latex,line width=0.2mm] (0,0.5) -- (0.5,0.5);
\draw[-latex,line width=0.2mm] (0.5,0.5) -- (1,0.5); 
\end{scope}

\begin{scope}[xshift=2cm,yshift=1cm]
\draw (0,0) rectangle (1,1);
\draw[latex-,line width=0.2mm] (0.5,0) -- (0.5,0.5);
\draw[-latex,line width=0.2mm] (0.5,0.5) -- (0.5,1); 
\draw[-latex,line width=0.2mm] (0,0.5) -- (0.5,0.5);
\draw[latex-,line width=0.2mm] (0.5,0.5) -- (1,0.5); 
\end{scope}

\begin{scope}[xshift=3cm,yshift=1cm]
\draw (0,0) rectangle (1,1);
\draw[latex-,line width=0.2mm] (0.5,0) -- (0.5,0.5);
\draw[latex-,line width=0.2mm] (0.5,0.5) -- (0.5,1); 
\draw[latex-,line width=0.2mm] (0,0.5) -- (0.5,0.5);
\draw[latex-,line width=0.2mm] (0.5,0.5) -- (1,0.5); 
\end{scope}

\begin{scope}[yshift=2cm]
\draw (0,0) rectangle (1,1);
\draw[latex-,line width=0.2mm] (0.5,0) -- (0.5,0.5);
\draw[-latex,line width=0.2mm] (0.5,0.5) -- (0.5,1); 
\draw[-latex,line width=0.2mm] (0,0.5) -- (0.5,0.5);
\draw[latex-,line width=0.2mm] (0.5,0.5) -- (1,0.5); 
\end{scope}

\begin{scope}[xshift=1cm,yshift=2cm]
\draw (0,0) rectangle (1,1);
\draw[-latex,line width=0.2mm] (0.5,0) -- (0.5,0.5);
\draw[latex-,line width=0.2mm] (0.5,0.5) -- (0.5,1); 
\draw[latex-,line width=0.2mm] (0,0.5) -- (0.5,0.5);
\draw[-latex,line width=0.2mm] (0.5,0.5) -- (1,0.5); 
\end{scope}

\begin{scope}[xshift=2cm,yshift=2cm]
\draw (0,0) rectangle (1,1);
\draw[-latex,line width=0.2mm] (0.5,0) -- (0.5,0.5);
\draw[-latex,line width=0.2mm] (0.5,0.5) -- (0.5,1); 
\draw[-latex,line width=0.2mm] (0,0.5) -- (0.5,0.5);
\draw[-latex,line width=0.2mm] (0.5,0.5) -- (1,0.5); 
\end{scope}

\begin{scope}[xshift=3cm,yshift=2cm]
\draw (0,0) rectangle (1,1);
\draw[latex-,line width=0.2mm] (0.5,0) -- (0.5,0.5);
\draw[-latex,line width=0.2mm] (0.5,0.5) -- (0.5,1); 
\draw[-latex,line width=0.2mm] (0,0.5) -- (0.5,0.5);
\draw[latex-,line width=0.2mm] (0.5,0.5) -- (1,0.5); 
\end{scope}

\begin{scope}[xshift=0cm,yshift=3cm]
\draw (0,0) rectangle (1,1);
\draw[-latex,line width=0.2mm] (0.5,0) -- (0.5,0.5);
\draw[-latex,line width=0.2mm] (0.5,0.5) -- (0.5,1); 
\draw[-latex,line width=0.2mm] (0,0.5) -- (0.5,0.5);
\draw[-latex,line width=0.2mm] (0.5,0.5) -- (1,0.5); 
\end{scope}

\begin{scope}[xshift=1cm,yshift=3cm]
\draw (0,0) rectangle (1,1);
\draw[latex-,line width=0.2mm] (0.5,0) -- (0.5,0.5);
\draw[-latex,line width=0.2mm] (0.5,0.5) -- (0.5,1); 
\draw[-latex,line width=0.2mm] (0,0.5) -- (0.5,0.5);
\draw[latex-,line width=0.2mm] (0.5,0.5) -- (1,0.5); 
\end{scope}

\begin{scope}[xshift=2cm,yshift=3cm]
\draw (0,0) rectangle (1,1);
\draw[-latex,line width=0.2mm] (0.5,0) -- (0.5,0.5);
\draw[-latex,line width=0.2mm] (0.5,0.5) -- (0.5,1); 
\draw[latex-,line width=0.2mm] (0,0.5) -- (0.5,0.5);
\draw[latex-,line width=0.2mm] (0.5,0.5) -- (1,0.5); 
\end{scope}

\begin{scope}[xshift=3cm,yshift=3cm]
\draw (0,0) rectangle (1,1);
\draw[-latex,line width=0.2mm] (0.5,0) -- (0.5,0.5);
\draw[-latex,line width=0.2mm] (0.5,0.5) -- (0.5,1); 
\draw[latex-,line width=0.2mm] (0,0.5) -- (0.5,0.5);
\draw[latex-,line width=0.2mm] (0.5,0.5) -- (1,0.5); 
\end{scope}

\end{tikzpicture}\]
\caption{\label{figure.six.vertex} An example of pattern that appears in an element of the six vertex model.}
\end{figure}

\subsubsection{\label{section.discrete.curves} Drifting discrete curves}

From the six vertex model, we derive another 
representation of square ice
through the application of the following invertible
transformation on the elements of $\mathcal{A}_0$:

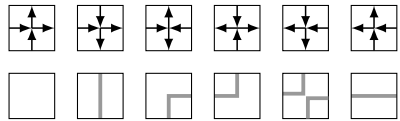
\[\begin{tikzpicture}[scale=0.6]
\begin{scope}
\draw (0,0) rectangle (1,1);
\draw[-latex,line width=0.2mm] (0.5,0) -- (0.5,0.5);
\draw[-latex,line width=0.2mm] (0.5,0.5) -- (0.5,1); 
\draw[-latex,line width=0.2mm] (0,0.5) -- (0.5,0.5);
\draw[-latex,line width=0.2mm] (0.5,0.5) -- (1,0.5);
\end{scope}

\begin{scope}[xshift=1.5cm]
\draw (0,0) rectangle (1,1);
\draw[latex-,line width=0.2mm] (0.5,0) -- (0.5,0.5);
\draw[latex-,line width=0.2mm] (0.5,0.5) -- (0.5,1); 
\draw[-latex,line width=0.2mm] (0,0.5) -- (0.5,0.5);
\draw[-latex,line width=0.2mm] (0.5,0.5) -- (1,0.5); 
\end{scope}

\begin{scope}[xshift=3cm]
\draw (0,0) rectangle (1,1);
\draw[latex-,line width=0.2mm] (0.5,0) -- (0.5,0.5);
\draw[-latex,line width=0.2mm] (0.5,0.5) -- (0.5,1); 
\draw[-latex,line width=0.2mm] (0,0.5) -- (0.5,0.5);
\draw[latex-,line width=0.2mm] (0.5,0.5) -- (1,0.5); 
\end{scope}

\begin{scope}[xshift=4.5cm]
\draw (0,0) rectangle (1,1);
\draw[-latex,line width=0.2mm] (0.5,0) -- (0.5,0.5);
\draw[latex-,line width=0.2mm] (0.5,0.5) -- (0.5,1); 
\draw[latex-,line width=0.2mm] (0,0.5) -- (0.5,0.5);
\draw[-latex,line width=0.2mm] (0.5,0.5) -- (1,0.5); 
\end{scope} 

\begin{scope}[xshift=6cm]
\draw (0,0) rectangle (1,1);
\draw[latex-,line width=0.2mm] (0.5,0) -- (0.5,0.5);
\draw[latex-,line width=0.2mm] (0.5,0.5) -- (0.5,1); 
\draw[latex-,line width=0.2mm] (0,0.5) -- (0.5,0.5);
\draw[latex-,line width=0.2mm] (0.5,0.5) -- (1,0.5); 
\end{scope}

\begin{scope}[xshift=7.5cm]
\draw (0,0) rectangle (1,1);
\draw[-latex,line width=0.2mm] (0.5,0) -- (0.5,0.5);
\draw[-latex,line width=0.2mm] (0.5,0.5) -- (0.5,1); 
\draw[latex-,line width=0.2mm] (0,0.5) -- (0.5,0.5);
\draw[latex-,line width=0.2mm] (0.5,0.5) -- (1,0.5); 
\end{scope}

\begin{scope}[yshift=-1.5cm]
\draw (0,0) rectangle (1,1);
\end{scope}

\begin{scope}[yshift=-1.5cm,xshift=1.5cm]
\draw[line width = 0.5mm,color=gray!80]
(0.5,0) -- (0.5,1);
\draw (0,0) rectangle (1,1);
\end{scope}

\begin{scope}[yshift=-1.5cm,xshift=3cm]
\draw[line width = 0.5mm,color=gray!80]
(0.5,0) -- (0.5,0.5) -- (1,0.5);
\draw (0,0) rectangle (1,1);
\end{scope}

\begin{scope}[yshift=-1.5cm,xshift=4.5cm]
\draw[line width = 0.5mm,color=gray!80]
(0,0.5) -- (0.5,0.5) -- (0.5,1);
\draw (0,0) rectangle (1,1);
\end{scope}

\begin{scope}[yshift=-1.5cm,xshift=6cm]
\draw[line width = 0.5mm,color=gray!80]
(0,0.55) -- (0.45,0.55) -- (0.45,1);
\draw (0,0) rectangle (1,1);
\draw[line width = 0.5mm,color=gray!80] (0.55,0) -- (0.55,0.45) -- (1,0.45);
\end{scope}

\begin{scope}[yshift=-1.5cm,xshift=7.5cm]
\draw[line width = 0.5mm,color=gray!80]
(0,0.5) -- (1,0.5);
\draw (0,0) rectangle (1,1);
\end{scope}
\end{tikzpicture}\]

The set of image symbols is denoted $\mathcal{A}_1$. 
The discrete drifting curves model 
is the subset, denoted $X^s$, of $\mathcal{A}_1^{\mathbb{Z}^2}$ 
that are obtained by the previous transformation 
of symbols.
The pattern of the six vertex model on Figure~\ref{figure.six.vertex} can 
be represented in the discrete drifting curves 
model as on Figure~\ref{figure.six.vertex.curves}. 

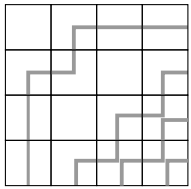
\begin{figure}[h!]
\[\begin{tikzpicture}[scale=0.6]
\begin{scope}
\draw[color=gray!80,line width=0.4mm] (0.5,0) -- 
(0.5,1);
\draw (0,0) rectangle (1,1);
\end{scope} 

\begin{scope}[xshift=1cm]
\draw[line width = 0.5mm,color=gray!80]
(0.55,0) -- (0.55,0.55) -- (1,0.55);
\draw (0,0) rectangle (1,1);
\end{scope}

\begin{scope}[xshift=2cm]
\draw[line width = 0.5mm,color=gray!80]
(0,0.55) -- (0.45,0.55) -- (0.45,1);
\draw (0,0) rectangle (1,1);
\draw[line width = 0.5mm,color=gray!80] (0.55,0) -- (0.55,0.55) -- (1,0.55);
\end{scope}
\begin{scope}[xshift=3cm]
\draw[line width = 0.5mm,color=gray!80]
(0,0.55) -- (0.45,0.55) -- (0.45,1);
\draw (0,0) rectangle (1,1);
\draw[line width = 0.5mm,color=gray!80] (0.55,0) -- (0.55,0.55) -- (1,0.55);
\end{scope}

\begin{scope}[yshift=1cm]
\draw[color=gray!80,line width=0.4mm] (0.5,0) -- 
(0.5,1);
\draw (0,0) rectangle (1,1);
\end{scope} 

\begin{scope}[xshift=1cm,yshift=1cm]
\draw (0,0) rectangle (1,1);
\end{scope}

\begin{scope}[xshift=2cm,yshift=1cm]
\draw[line width = 0.5mm,color=gray!80]
(0.45,0) -- (0.45,0.55) -- (1,0.55);
\draw (0,0) rectangle (1,1);
\end{scope}

\begin{scope}[xshift=3cm,yshift=1cm]
\draw[line width = 0.5mm,color=gray!80]
(0,0.55) -- (0.45,0.55) -- (0.45,1);
\draw (0,0) rectangle (1,1);
\draw[line width = 0.5mm,color=gray!80] (0.45,0) -- (0.45,0.45) -- (1,0.45);
\end{scope}

\begin{scope}[xshift=0cm,yshift=2cm]
\draw[line width = 0.5mm,color=gray!80]
(0.5,0) -- (0.5,0.5) -- (1,0.5);
\draw (0,0) rectangle (1,1);
\end{scope}

\begin{scope}[xshift=1cm,yshift=2cm]
\draw[line width = 0.5mm,color=gray!80]
(0,0.5) -- (0.5,0.5) -- (0.5,1);
\draw (0,0) rectangle (1,1);
\end{scope}

\begin{scope}[xshift=2cm,yshift=2cm]
\draw (0,0) rectangle (1,1);
\end{scope}

\begin{scope}[xshift=3cm,yshift=2cm]
\draw[line width = 0.5mm,color=gray!80]
(0.45,0) -- (0.45,0.5) -- (1,0.5);
\draw (0,0) rectangle (1,1);
\end{scope}

\begin{scope}[xshift=0cm,yshift=3cm]
\draw (0,0) rectangle (1,1);
\end{scope}

\begin{scope}[xshift=1cm,yshift=3cm]
\draw[line width = 0.5mm,color=gray!80]
(0.5,0) -- (0.5,0.5) -- (1,0.5);
\draw (0,0) rectangle (1,1);
\end{scope}

\begin{scope}[xshift=2cm,yshift=3cm]
\draw[line width = 0.5mm,color=gray!80]
(0,0.5) -- (1,0.5);
\draw (0,0) rectangle (1,1);
\end{scope}

\begin{scope}[xshift=3cm,yshift=3cm]
\draw[line width = 0.5mm,color=gray!80]
(0,0.5) -- (1,0.5);
\draw (0,0) rectangle (1,1);
\end{scope}

\end{tikzpicture}\]
\caption{\label{figure.six.vertex.curves} 
Representation of pattern on Figure~\ref{figure.six.vertex}.}
\end{figure}

\subsection{\label{section.lieb.path} Lieb path of transfer matrices}

A pattern of $X^s$ is an element 
of $\mathcal{A}_1^{\mathbb{U}}$ for some 
finite subset $\mathbb{U}$ of $\mathbb{Z}^2$ 
which appears in an element of $X^s$.
A $(N,1)$-cylindric pattern of $X^s$ is a 
pattern in $\mathcal{A}_1^{\{1,...,N\}\times \{0\}}$ that can be wrapped on the cylinder 
$\mathcal{A}_1^{\mathbb{Z}/n\mathbb{Z}}$ 
without breaking the rules of $X^s$. Let us 
denote $\{0,1\}^N_{*}$ the set of length $N$ 
words on the alphabet $\{0,1\}$.
For a word $\boldsymbol{\epsilon} \in \{0,1\}^N_{*}$, 
we denote $|\boldsymbol{\epsilon}|_1$ the number 
of integers $k \in \{1,...,N\}$ such that $\boldsymbol{\epsilon}_k = 1$.

\begin{notation}
Let $N \ge 1$ be an integer. Let 
us denote $\Omega_N$ the space $\mathbb{C}^2 \bigotimes ... \bigotimes \mathbb{C}^2$, 
tensor product of $N$ copies 
of $\mathbb{C}^2$, whose canonical basis elements 
are denoted indifferently by $\boldsymbol{\epsilon}= \ket{\boldsymbol{\epsilon}_1 ... \boldsymbol{\epsilon}_N}$ or the words 
$\boldsymbol{\epsilon}_1 ... \boldsymbol{\epsilon}_N$,
for $(\boldsymbol{\epsilon}_1, ..., \boldsymbol{\epsilon}_N) \in \{0,1\}^N$, 
according to quantum mechanics 
notations, in order to distinguish them from 
the coordinate definition of 
vectors of $\Omega_N$. We also denote by $\Omega_N^{(n)}$ 
the vector space generated by the elements $\boldsymbol{\epsilon}$ of 
the canonical basis such that $|\boldsymbol{\epsilon}|_1 = n$.
\end{notation}

\begin{notation}
For any $\boldsymbol{\epsilon}$ in the canonical 
basis of $\Omega_N$ such that 
$|\boldsymbol{\epsilon}|_1 = n$, 
we denote by $q_k [\boldsymbol{\epsilon}]$ the $k$th 
position $j$ in $\{1,...,N\}$ such that $\boldsymbol{\epsilon}_j=1$.
\end{notation}

\begin{notation}
For all $N$ and $(N,1)$-cylindric
pattern $w$, let $|w|$ denote the number 
of symbols 
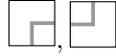
\[\begin{tikzpicture}[scale=0.6]
\begin{scope}[yshift=-1.5cm,xshift=3cm]
\draw[line width = 0.5mm,color=gray!80]
(0.5,0) -- (0.5,0.5) -- (1,0.5);
\draw (0,0) rectangle (1,1);
\end{scope}
\end{tikzpicture}, \begin{tikzpicture}[scale=0.6]
\begin{scope}[yshift=-1.5cm,xshift=3cm]
\draw[line width = 0.5mm,color=gray!80]
(0,0.5) -- (0.5,0.5) -- (0.5,1);
\draw (0,0) rectangle (1,1);
\end{scope}
\end{tikzpicture}\]
in this pattern. Such a pattern is said to link some 
$\boldsymbol{\epsilon} \in \{0,1\}^N_*$ to 
$\boldsymbol{\eta} \in \{0,1\}^N_*$ when there is a 
curbe entering through the south (resp. outgoing
through the north side) of $w$ at position $k$ 
if and only if $\boldsymbol{\epsilon}_k =1$ (resp. $\boldsymbol{\epsilon}_k =0$). This relation is denoted 
$\boldsymbol{\epsilon} \mathcal{R}[w] 
\boldsymbol{\eta}$. 
\end{notation}

\begin{definition}
For all $t \ge 0$, $V_N (t) \in \mathcal{M}_{2^N} (\mathbb{C})$ denotes
the matrix such that for all $\boldsymbol{\epsilon},\boldsymbol{\eta} \in \{0,1\}^{*}_N$, 

\[V_N (t) [\boldsymbol{\epsilon},\boldsymbol{\eta}] =\displaystyle{\sum_{\boldsymbol{\epsilon} \mathcal{R}[w] \boldsymbol{\eta}}} t^{|w|}\]
\end{definition}

Let us notice that for all $N$, the map $t \mapsto V_N (t)$ 
is analytic, which we called Lieb path of transfer matrices 
in~\cite{Gangloff2019}.

\subsection{\label{section.coordinate.bethe.ansatz} Coordinate Bethe ansatz}

Since the entropy of $X^s$ is expressed 
according to the greatest 
egeinvalues on the subspaces $\Omega_N^{(n)}$
of $V_N(1)$, for all $N$, the strategy 
in order to compute the entropy is 
to search for an expression of these eigenvalues. 
This is the purpose of the \textbf{coordinate 
Bethe ansatz}, which produces 
candidate eigenvalues. Some auxiliary functions 
involved in the method depend on the position 
of $t$ according to $2$. Since the method is 
independant provided the auxiliary functions, 
we present it only for the case $t <2$, 
which corresponds to our interest.
In the proof of the value of square ice entropy, 
these eigenvalues are identified 
with the maximal ones on a subdomain of the parameter $t$, and this information is 
transported to the parameter $1$ through 
analyticity.

\subsubsection{\label{section.notations} Auxiliary functions} 

Let us denote $\mu : (-1,1) \rightarrow (0,\pi)$ the inverse of the function $\cos : 
(0,\pi) \rightarrow (-1,1)$.
For all $t \in [0,2)$, we will denote 
$\Delta_t = \frac{2-t^2}{2}$, $\mu_t 
= \mu(-\Delta_t)$, and $I_t = (-(\pi-\mu_t), 
(\pi-\mu_t))$. 

\begin{notation}
Let us denote 
$\Theta$ the unique analytic function $(t,x,y) 
\mapsto \Theta_t (x,y)$ 
from the set $\{(t,x,y): t \in (0,2), x,y \in I_t\}$ 
to $\mathbb{R}$ such that $\Theta_{\sqrt{2}} 
(0,0) = 0$ and for all $t,x,y$,

\[\exp(-i\Theta_t (x,y)) = \exp(i(x-y)). \frac{e^{-ix} + e^{iy} -  2 \Delta_t}{e^{-iy} + e^{ix} -  2 \Delta_t}.\]
\end{notation} 

By a unicity argument, for all $t,x,y$, $\Theta_t (x,y) = - \Theta_t (-x,-y)$.

\begin{notation}
Let us denote $\kappa$ the 
unique analytic map $(t,\alpha) \mapsto \kappa_t (\alpha)$ 
from $(0,2) \times \mathbb{R}$ to $\mathbb{R}$
such that $\kappa_{\sqrt{2}} (0)=0$ 
and for all $t,\alpha$,
\[e^{i\kappa_t (\alpha)} = \frac{e^{i\mu_t}-e^{\alpha}}{e^{i\mu_t+\alpha} - 1}.\]
\end{notation}

\begin{computation}
\label{computation.kappa.derivative}
For all $t \in (0,2)$, and $\alpha \in \mathbb{R}$, 
the derivative of $\kappa_t$ in $\alpha$ is given by: 
\[\kappa'_t (\alpha) = \frac{\sin(\mu_t)}{\cosh(\alpha)-\cos(\mu_t)}.\]
\end{computation}

\begin{proof}

\begin{itemize}
\item \textbf{Computation of $\cos(\kappa_t(\alpha))$ and $sin(\kappa_t(\alpha))$:}

\[e^{i\kappa_t (\alpha)} = 
 \frac{\left(e^{-i\mu_t +\alpha}-1\right)\left(e^{i\mu_t}-e^{\alpha}\right)}{\left|e^{i\mu_t+\alpha}-1\right|^2} = \frac{e^{\alpha}
+ e^{2\alpha} e^{-i\mu_t} - e^{i\mu_t} + e^{\alpha}}{(\cos(\mu_t)e^{\alpha}-1)^2+(\sin(\mu_t)e^{\alpha})^2}.\]

Thus by taking the real part,
\[\cos(\kappa_t(\alpha)) = \frac{2e^{\alpha} 
+ (e^{2\alpha}-1)\cos(\mu_t)}{\cos^2(\mu_t)e^{2\alpha} - 2\cos(\mu_t) e^{\alpha} + 1 + (1- \cos^2 (\mu_t))e^{2\alpha}}
\]

\[\cos(\kappa_t(\alpha)) = \frac{2e^{\alpha} 
+ (e^{2\alpha}-1)\cos(\mu_t)}{e^{2\alpha} - 2\cos(\mu_t) e^{\alpha} + 1}
= \frac{1
- \cos(\mu_t)\cosh(\alpha)}{\cosh(\alpha)-\cos(\mu_t)},\]
where we factorized by $2e^{\alpha}$ for
the second equality. As a consequence:

\[\cos(\kappa_t(\alpha)) =\frac{\sin^2(\mu_t)+\cos^2(\mu_t)-\cos(\mu_t) \cosh(\alpha)}{\cosh(\alpha)-\cos(\mu_t)} = \frac{\sin^2(\mu_t)}{\cosh(\alpha)-\cos(\mu_t)} 
- \cos(\mu_t).\]

A similar computation gives 

\[\sin(\kappa_t(\alpha)) = \frac{\sin(\mu_t) \sinh(\alpha)}{\cosh(\alpha)-\cos(\mu_t)}\]

\item \textbf{Deriving the expression $\cos(\kappa_t(\alpha))$:}

As a consequence, for all $\alpha$: 

\[-\kappa'_t(\alpha) \sin(\kappa_t(\alpha)) = - \frac{\sin^2 (\mu_t) \sinh(\alpha)}{(\cosh(\alpha)-\cos(\mu_t))^2} = -\frac{\sin(\kappa_t(\alpha))^2}{\sinh(\alpha)}.\]
\[\]

Thus, for all $\alpha$ but in 
a discrete subset of 
$\mathbb{R}$, 
\[\kappa'_t(\alpha) = \frac{\sin(\mu_t)}{\cosh(\alpha)-\cos(\mu_t)}.\]
This identity is thus verified 
on all $\mathbb{R}$, by continuity.
\end{itemize}

\end{proof}

From this, one deduces that $\kappa_t$ is increasing. 
By another unicity argument, this function is 
antisymmetric: for all $\alpha$, $\kappa_t (-\alpha) = -\kappa_t (\alpha)$. It is also known that $\kappa_t : \mathbb{R} \rightarrow I_t$ is an invertible map (see for instance ~\cite{Gangloff2019}, Proposition 7). Moreover: 

\begin{computation}
For all $t \in (0,2)$ and $\alpha \in I_t$,
\[\kappa'_t (\kappa^{-1}(\alpha)) = \frac{\cos(\alpha) + \cos(\mu_t)}{\sin(\mu_t)}\]
\end{computation}

\begin{proof}
From the first point of the proof of Computation~\ref{computation.kappa.derivative}, 
\[\cos(\kappa_t (\alpha)) = \frac{\sin^2(\mu_t)}{\cosh(\alpha)-\cos(\mu_t)} 
- \cos(\mu_t) = \sin(\mu_t) \kappa'_t (\alpha) - \cos(\mu_t).\]
As a consequence, 
\[\kappa'_t (\kappa_t^{-1} (\alpha)) = \frac{\cos(\alpha) + \cos(\mu_t)}{\sin(\mu_t)}.\]
\end{proof}

 \subsubsection{Statement of the ansatz}

 \begin{notation}
 \label{notation.vector.ansatz}
For all $(p_1,...,p_n) \in I_t^n$, let us denote $\psi_{t,n,N} (p_1,...,p_n)$ the 
 vector in $\Omega_N$ such that for all 
 $\boldsymbol{\epsilon} \in \{0,1\}^*_N$, 
 \[\psi_{t,n,N}(p_1,...,p_n)[\boldsymbol{\epsilon}]
 = \sum_{\sigma \in \Sigma_n} C_{\sigma} (t) [p_1,...,p_n] \prod_{k=1}^n e^{ip_{\sigma(k)} q_k [\boldsymbol{\epsilon}]},\]
 where (for $\epsilon(\sigma)$ denoting the 
 signature of $\sigma$):
\[C_{\sigma} (t) [p_1,...,p_n] = \epsilon(\sigma)\prod_{1 \le k <l\le n} e^{ip_{\sigma(k)}} \left( e^{-ip_{\sigma(k)}} + e^{ip_{\sigma(l)}} - 2\Delta_t \right).\]
 \end{notation}
 
\begin{notation}
For all $t$ and $z \neq 1$, we set 
\[L_t (z) = 1 + \frac{t^2 z}{1-z}, \qquad 
M_t (z) = 1- \frac{t^2}{1-z}.\]
\end{notation} 
 
 \begin{notation}
 Let $(p_1,...,p_n) \in I_t^n$ 
 such that $p_1 < ... < p_n$. 
If for all $j$, $p_j \neq 0$, we denote 
\[\Lambda_{n,N} (t) [p_1,...,p_n] = \prod_{k=1}^n L_t (e^{ip_k}) + \prod_{k=1}^n M_t (e^{ip_k}).\]
If there exists some $l$ such that $p_l=0$:
\[\Lambda_{n,N} (t) [p_1,...,p_n] = \left( 2 + t^2 (N-1) + \sum_{k \neq l} 
\frac{\partial \Theta_t}{\partial x} (0,p_k)\right)\prod_{k=1}^n M_t (e^{ip_k}).\]
 \end{notation}
 
 \begin{theorem}
 \label{theorem.coordinate.ansatz}
 For all $N$ and $n \le N/2$, and $(p_1,...,p_n) \in I_t^n$ such that $p_1 < ... < p_n$
 and for all $j$ the following equation is verified:
 
 \[(E_j)[t,n,N]: \qquad Np_j = 2\pi j - (n+1)\pi - \sum_{k=1}^n \Theta_t (p_j,p_k).\]
 Then we have: 
 \[V_N (t).\psi_{t,n,N} (p_1,...,p_n) = \Lambda_{n,N} (t) [p_1,...,p_n] \psi_{t,n,N} (p_1,...,p_n).\]
 \end{theorem}

\section{\label{section.overview} Overview}

In the following, we provide a proof 
of Theorem~\ref{theorem.coordinate.ansatz}
using the algebraic Bethe ansatz. 
In Section~\ref{section.yang.baxter.paths}, we expose an abstract version of it, 
which consists, given an analytic path 
of \textit{commuting} transfer matrices that we call 
\textbf{Yang-Baxter path} and denoted 
$x \mapsto T_N^x$ (commuting means 
that for all $x,y$, 
$T_N^x$ and $T_N^y$ commute) in constructing, 
under the existence of a sequence $(x_1,...,x_n)$ 
which verifies a system of non-linear equations,
a candidate eigenvector and eigenvalue 
for any matrix $T^x$, using the commutation 
relations between the transfer matrices.
These transfer matrices are constructed from 
the iteration of local matrices.
We describe, in Section~\ref{section.algebraic.bethe.ansatz}, the construction of a Yang-Baxter
path provided a solution of the so-called Yang-Baxter 
equation on these local matrices.
In Section~\ref{section.trigonometric.path}, for 
all $t$, we apply the algebraic Bethe 
ansatz to a trigonometric Yang-Baxter path,
 $x \mapsto T_N^x (t)$, which takes the 
 value $V_N(t)$ for some parameter, 
 and give a proof of Theorem~\ref{theorem.coordinate.ansatz}.
 See an illustration on Figure~\ref{figure.strategy.yang.baxter}.
 In the end, we derive, using the Yang-Baxter path of commuting matrices, the commutation 
 of $V_N (t)$ with some Heisenberg Hamiltonian [Section~\ref{section.commutation}].
In the proof of the value of square ice, this property is used 
to identify the maximal eigenvalue of $V_N(t)$.

\begin{figure}[h!]
\[\begin{tikzpicture}[scale=0.5]

\draw[dashed] (4,3) arc (180:260:2cm);
\draw[gray!90,latex-,line width = 0.4mm] (4,3) arc (180:220:2cm);
\draw[dashed] (4,3) arc (0:80:2cm);
\draw[gray!90,latex-,line width = 0.4mm] (4,3) arc (0:40:2cm);
\draw[red] (0,0) arc (180:90:3cm);

\draw[red] (3,3) -- (5,3);
\draw[red] (5,3) arc (-90:0:3cm);
\draw[red] (0,0) -- (0,-1);
\fill[red] (0,0) circle (1.25mm);
\fill[gray!90] (3.5,4.3) circle (1.25mm);
\fill[gray!90] (4.5,1.7) circle (1.25mm);
\fill[red] (4,3) circle (1.25mm);
\node at (-2,0) {$V_{N} (1)$};
\node at (3,2) {$V_{N} (t)$};
\node at (9,1) {Yang-Baxter path};
\node[scale=0.8,color=gray!90] at (1,3.5) {Alg. Bethe ansatz};
\node at (10,6) {Lieb path};
\node at (4.5,5.3) {$T^x_{N} (t)$};
\end{tikzpicture}\]
\caption{\label{figure.strategy.yang.baxter} Illustration of the strategy of proof 
for Theorem~\ref{theorem.coordinate.ansatz}.}
\end{figure}
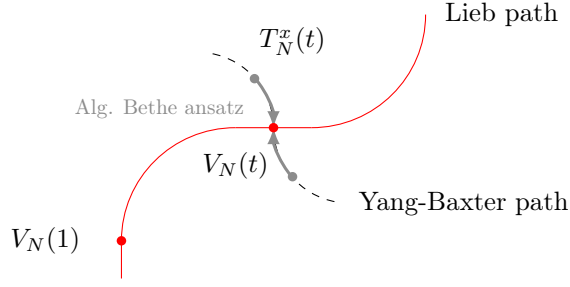

\section{\label{section.yang.baxter.paths} 
Construction of Yang-Baxter paths}

In this section, we explain how to construct 
more general transfer matrix [Section~\ref{section.from.local.to.transfer}], and thus 
paths, deriving from local matrix functions [Section~\ref{section.local.matrix.functions}]. The interest of this construction of 
transfer matrices is that it 
allows a simple criterion 
for these matrices to commute, widely known 
as the Yang-Baxter equation [Section~\ref{section.yang.baxter.equation}]. This criterion becomes 
simpler when the consider 
local matrix functions are strongly 
symmetric [Section~\ref{section.strongly.symmetric}]: 
it consists in selecting the local matrix 
functions in the same level surface 
of an operator on these functions. This 
simpler version helps to find some 
Yang-Baxter paths. We define the notion of Yang-Baxter path in Section~\ref{subsection.yang.baxter.paths}. \bigskip

In the following, for any finite sequence of matrices $M^{(i)}$ having the same size, we denote by \[\prod_{k=1}^m M^{(k)} = M^{(1)} \times ... \times M^{(m)}.\]

\subsection{\label{section.local.matrix.functions} Local matrix functions}

\begin{definition} \label{definition.local.matrix}
A \textbf{local matrix function} of square ice
is a function $R$ from $\{0,1\}^2$ to $\mathcal{M}_2 (\mathbb{C})$ such that for all $u,v \in \{0,1\}$, 
$(w,w') \in \{0,1\}^2$, if 
there is no symbol of the discrete 
curves shift $X^s$ such that there 
is a curve crossing the west (resp. north, south, east) boundary if 
and only if $w'$ (resp. $v$, $u$, $w$) equals 1
(see an illustration on Figure~\ref{figure.definition.R.matrix}), 
then
\[R(u,v)[w,w'] = 0.\]
An image of the local matrix function is
called a \textbf{local matrix}. The \textbf{
matricial representation} of this 
local matrix function is the element of $\mathcal{M}_4 (\mathbb{C})$ defined as:

\[\left( \begin{array}{cc} R(0,0) & R(0,1) \\ 
R(1,0) & R(1,1) \end{array} 
\right) = \left( \begin{array}{cccc}
R(0,0)[0,0] & R(0,0)[0,1] & R(0,1)[0,0] & R(0,1)[0,1] \\
R(0,0)[1,0] & R(0,0)[1,1] & R(0,1)[1,0] & R(0,1)[1,1] \\
R(1,0)[0,0] & R(1,0)[0,1] & R(1,1)[0,0] & R(1,1)[0,1] \\
R(1,0)[1,0] & R(1,0)[0,1] & R(1,1)[1,0] & R(1,1)[1,1] \\
\end{array}\right).\]
\end{definition}

\begin{figure}[h!]

\[
\begin{tikzpicture}
\draw (0,0) rectangle (1,1);
\draw[line width=0.6mm,color=gray!80, dashed] (-0.15,0.5) -- (0.15,0.5);
\draw[line width=0.6mm,color=gray!80, dashed] (0.85,0.5) -- (1.15,0.5);
\draw[line width=0.6mm,color=gray!80, dashed] (0.5,-0.15) -- (0.5,0.15);
\draw[line width=0.6mm,color=gray!80, dashed] (0.5,0.85) -- (0.5,1.15);
\node at (1.4,0.5) {$w'$};
\node at (-0.4,0.5) {$w$};
\node at (0.5,1.4) {$v$};
\node at (0.5,-0.4) {$u$};
\node at (-1.6,0.5) {input};
\node at (2.6,0.5) {output};
\end{tikzpicture}\]

\caption{\label{figure.definition.R.matrix} Definition of non-zero coefficients 
in a local matrix; the dashed segments designate possible curve crossing the border of the 
symbol.}
\end{figure}
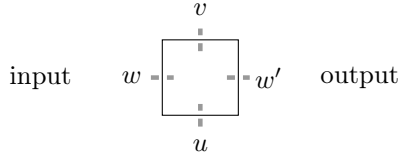

\begin{fact} \label{fact.local.matrix}
A direct translation of 
Definition~\ref{definition.local.matrix} is that when 
$u=v$, the matrix $R(u,v)$ is diagonal, and 
$R(0,1)$ and $R(1,0)$ have all coefficients equal to 
zero except for the respective coefficients 
$R(0,1)[1,0]$
and $R(1,0)[0,1]$. As a consequence, the local matrix 
function $R$ can be represented as 
\[R= \left( \begin{array}{cccc}
a & 0 & 0 & 0 \\
0 & b & c & 0 \\
0 & d & e & 0 \\
0 & 0 & 0 & f \\
\end{array}\right),\]
for some $a,b,c,d,e,f \in \mathbb{C}$.
\end{fact}

\begin{remark}
Although this notion is defined as a function 
having matrix values, this notion is usually  designated by the term of $R$-matrix, which 
refers more precisely to the representation 
of the local matrix function. 
\end{remark}

\begin{definition}
A \textbf{north-east local matrix} is a 
matrix $ Q \in \mathcal{M}_4 (\mathbb{C})$
seen as the matrix of a linear operation
from $\mathbb{C}^2 \otimes \mathbb{C}^2$ 
to itself, such that for all $t,u,v,w \in \{0,1\}$, if there is no symbol of the discrete 
curves shift $X^s$ such that there 
is a curve crossing the west (resp. south, north, east) boundary if 
and only if $u$ (resp. $t$, $w$, $v$) equals 1
(see an illustration on Figure~\ref{figure.definition.north.east.matrix}), 
then
\[Q[(t,u),(v,w)] = 0.\]
The matrix $Q$ is then equal to: 
\[\left( \begin{array}{cccc}
Q[(0,0),(0,0)] & Q[(0,0),(0,1)] & Q[(0,0),(1,0)] & Q[(0,0),(1,1)] \\
Q[(0,1),(0,0)] & Q[(0,1),(0,1)] & Q[(0,1),(1,0)] & Q[(0,1),(1,1)] \\
Q[(1,0),(0,0)] & Q[(1,0),(0,1)] & Q[(1,0),(1,0)] & Q[(1,0),(1,1)] \\
Q[(1,1),(0,0)] & Q[(1,1),(0,1)] & Q[(1,1),(1,0)] & Q[(1,1),(1,1)] \\
\end{array}\right).\]
This means that the columns are indexed, 
from left to right (resp. from top to bottom) 
by the sequences 
$(0,0), (0,1),(1,0)$ and $(1,1)$.
\end{definition}

\begin{figure}[h!]

\[
\begin{tikzpicture}
\draw (0,0) rectangle (1,1);
\draw[line width=0.6mm,color=gray!80, dashed] (-0.15,0.5) -- (0.15,0.5);
\draw[line width=0.6mm,color=gray!80, dashed] (0.85,0.5) -- (1.15,0.5);
\draw[line width=0.6mm,color=gray!80, dashed] (0.5,-0.15) -- (0.5,0.15);
\draw[line width=0.6mm,color=gray!80, dashed] (0.5,0.85) -- (0.5,1.15);
\node at (1.4,0.5) {$v$};
\node at (-0.4,0.5) {$u$};
\node at (0.5,1.4) {$w$};
\draw[latex-latex] (1.6,0.5) -- (2,0.5) -- (2,2) -- (0.5,2) -- (0.5,1.6);
\draw (2,1) -- (2.2,1);
\node at (0.5,-0.4) {$t$};
\node at (-1.8,0) {input};
\draw[-latex] (-1,0.5) -- (-0.6,0.5);

\draw[-latex] (-1,0.5) -- (-1,-1) -- (0.5,-1) -- (0.5,-0.6);
\draw (-1.2,0) -- (-1,0);

\node at (2.8,1) {output};
\end{tikzpicture}\]

\caption{\label{figure.definition.north.east.matrix} Definition of non-zero coefficients 
in a north east local matrix; the dashed segments designate possible curve crossing the border of the 
symbol.}
\end{figure}
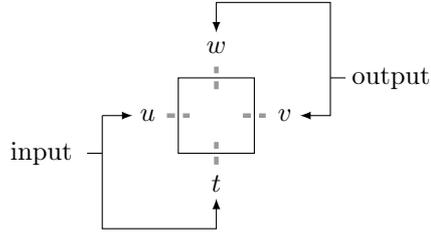

\begin{property}
\label{property.north.east.local.matrix}
From a local matrix function $R$, it is 
possible to derive a north east local matrix
$Q$ as follows: for all $t,u,v,w$,
\[Q[(t,u),(v,w)] = R(t,w)[u,v].\] 
This association is bijective.
As a consequence, if $R$ is represented as 
\[R= \left( \begin{array}{cccc}
a & 0 & 0 & 0 \\
0 & b & c & 0 \\
0 & d & e & 0 \\
0 & 0 & 0 & f \\
\end{array}\right),\]
then $Q$ is equal to :

\[Q= \left( \begin{array}{cccc}
R(0,0)[0,0] & R(0,1)[0,0] & R(0,0)[0,1] & R(0,1)[0,1] \\
R(0,0)[0,1] & R(0,1)[1,0] & R(0,0)[1,1] & R(0,1)[1,1] \\
R(1,0)[0,0] & R(1,1)[0,0] & R(1,0)[0,1] & 
R(1,1)[0,1] \\
R(1,0)[1,0] & R(1,1)[1,0] & R(1,0)[1,1] & R(1,1)[1,1] \\
\end{array}\right) = 
\left( \begin{array}{cccc}
a & 0 & 0 & 0 \\
0 & c & b & 0 \\
0 & e & d & 0 \\
0 & 0 & 0 & f \\
\end{array}\right)\]
\end{property}

\subsection{\label{section.from.local.to.transfer} Derivation of a transfer matrix
from a local matrix function}

Let $R$ be a local matrix function of square ice.

\begin{definition}
For all $N$, the $N$th 
\textbf{monodromy matrix} 
relative to 
$\vec{u},\vec{v} \in \{0,1\}^{*}_N$ is the matrix 
$M_N (\vec{u},\vec{v}) \in \mathcal{M}_2 (\mathbb{C})$ such 
that 

\[M_N (\vec{u},\vec{v}) = \left(\prod_{k=1}^{N} R(\vec{u}_k,\vec{v}_k)\right),\]
meaning that for all 
$(w,w') \in \{0,1\}^2$ 
\begin{align*}
M_N (\vec{u},\vec{v}) [w,w'] & = \sum_{w_2 \in \{0,1\}} ... \sum_{w_{N+1} \in \{0,1\}} 
\prod_{k=1}^{N} R(\vec{u}_k,\vec{v}_k) [w_k,w_{k+1}]
\end{align*}
where we denote $w_1 = w$ and $w_{N+1} = w'$.
See an illustration on Figure~\ref{figure.definition.monodromy.matrix}.
The matrix $M_N (\vec{u},\vec{v})$ is thus represented by : 
\[M_N (\vec{u},\vec{v}) = \left( \begin{array}{cc} 
M_N (\vec{u},\vec{v}) [0,0] & M_N (\vec{u},\vec{v}) [0,1]\\
M_N (\vec{u},\vec{v}) [1,0] & M_N (\vec{u},\vec{v}) [1,1]
\end{array}\right).\]
This representation is usually called the \textbf{Lax matrix}.
\end{definition}

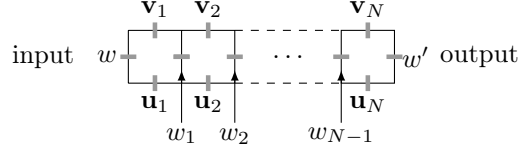
\begin{figure}[h!]

\[
\begin{tikzpicture}[scale=0.7]
\draw (0,0) grid (2,1);

\node at (3,0.5) {$\hdots$};

\draw (4,0) rectangle (5,1);
\draw[line width=0.6mm,color=gray!80] (-0.15,0.5) -- (0.15,0.5);
\draw[line width=0.6mm,color=gray!80] (0.85,0.5) -- (1.15,0.5);
\draw[line width=0.6mm,color=gray!80] (0.5,-0.15) -- (0.5,0.15);
\draw[line width=0.6mm,color=gray!80] (0.5,0.85) -- (0.5,1.15);
\node at (1,-1) {$w_1$};
\draw[-latex] (1,-0.75) -- (1,0.25);
\node at (-0.4,0.5) {$w$};
\node at (0.5,1.4) {$\vec{v}_1$};
\node at (0.5,-0.4) {$\vec{u}_1$};
\node at (-1.6,0.5) {input};

\begin{scope}[xshift=4cm]
\node at (0,-1) {$w_{N-1}$};
\draw[-latex] (0,-0.75) -- (0,0.25);
\node at (0.5,1.4) {$\vec{v}_N$};
\node at (0.5,-0.4) {$\vec{u}_N$};
\node at (1.4,0.5) {$w'$};
\draw[line width=0.6mm,color=gray!80] (-0.15,0.5) -- (0.15,0.5);
\draw[line width=0.6mm,color=gray!80] (0.85,0.5) -- (1.15,0.5);
\draw[line width=0.6mm,color=gray!80] (0.5,-0.15) -- (0.5,0.15);
\draw[line width=0.6mm,color=gray!80] (0.5,0.85) -- (0.5,1.15);

\node at (2.6,0.5) {output};
\end{scope}

\begin{scope}[xshift=1cm]
\node at (1,-1) {$w_2$};
\draw[-latex] (1,-0.75) -- (1,0.25);
\node at (0.5,1.4) {$\vec{v}_2$};
\node at (0.5,-0.4) {$\vec{u}_2$};
\draw[line width=0.6mm,color=gray!80] (-0.15,0.5) -- (0.15,0.5);
\draw[line width=0.6mm,color=gray!80] (0.85,0.5) -- (1.15,0.5);
\draw[line width=0.6mm,color=gray!80] (0.5,-0.15) -- (0.5,0.15);
\draw[line width=0.6mm,color=gray!80] (0.5,0.85) -- (0.5,1.15);
\end{scope}

\draw[dashed] (2,1) -- (4,1);
\draw[dashed] (2,0) -- (4,0);

\end{tikzpicture}\]

\caption{\label{figure.definition.monodromy.matrix} Illustration of the definition of the monodromy 
matrix relative to 
$\vec{u},\vec{v}$ defined from a local matrix. The coefficient 
$w,w'$ of the matrix is the sum 
over all possible compositions of 
local matrices, from the leftmost square to 
the rightmost one.}
\end{figure}

\begin{definition}
For all $N$, the 
$N$th \textbf{transfer matrix} associated to 
the local matrix function $R$ is the matrix 
$T_N \in \mathcal{M}_{2^N} (\mathbb{C})$, 
thought as 
a matrix of an operator on $\mathbb{C}^2 
\otimes .. \otimes \mathbb{C}^2$, 
such that for all $\vec{u},\vec{v} \in 
\{0,1\}^{*}_N$, 
\[T_N [\vec{u},\vec{v}] = \text{Tr}(M_N (\vec{u},\vec{v})).\] 
\end{definition}

\begin{remark}
The fact that the entropy of $X^s_N$ 
is equal 
to the entropy of $\overline{X}^s_N$ is involved 
in this definition through the fact that the 
transfer matrix is defined as the trace 
of the monodromy matrix, which is crucial 
in the proof of the commutation criterion.
\end{remark}

\begin{example}
For instance, the transfer matrix $V_N (t)$
 is obtained 
from the local matrix function represented as:
\[\left( \begin{array}{cccc}
1 & 0 & 0 & 0 \\
0 & 1 & t & 0 \\
0 & t & 1 & 0 \\
0 & 0 & 0 & 1
\end{array}\right).\]
\end{example}

\subsection{\label{section.yang.baxter.equation} Yang-Baxter equation, commutation criterion for transfer matrices}

For two matrices $M, M' \in \mathcal{M}_2 (\mathbb{C})$, we denote 
\[M \otimes M' = \left( \begin{array}{cc}
M[0,0].M' & M[0,1].M' \\
M[1,0].M' & M[1,1].M'
\end{array}\right)\]
the Kronecker product of $M$ and $M'$. 
We also denote $\overline{M}= \left( \begin{array}{cc} M & 0 \\
0 & M \end{array} \right)$.
We use similar notations for matrices 
in $\mathcal{M}_n (\mathbb{C})$ for 
all $n \ge 1$.

\begin{notation}
Let us consider $R$ and $R'$ two local 
matrix functions. For all $t,v \in \{0,1\}$, 
we denote 
\[R' \circ R (t,v) = \sum_{u \in \{0,1\}} 
R(t,u) \otimes R'(u,v),\] 
which can be depicted as on Figure~\ref{figure.product.local.matrices}. 
\end{notation}

\begin{notation}
For two local matrix functions $R,R'$, 
we denote $M_N \otimes M'_N (\vec{t},\vec{v})$ the 
matrix in $\mathcal{M}_4 (\mathbb{C})$ such that: 

\[M_N \otimes M'_N (\vec{t},\vec{v}) = \prod_{k=1}^{N} 
R' \circ R (\vec{t}_k,\vec{v}_k).\]
where $M_N$ and $M'_N$ are the respective $N$th 
monodromy matrices of $R$ and $R'$. 
\end{notation}

\begin{lemma}
\label{lemma.product.monodromy.matrices}
Let $R$ and $R'$ be two local matrix functions and 
$T_N$, $T'_N$ their respective transfer matrices. 
Then for all $\vec{t},\vec{v} \in \{0,1\}^{*}_N$:
\[(T_N.T'_N) [\vec{t},\vec{v}] = \text{Tr} (M_N \otimes M'_N (\vec{t},\vec{v})).\]
\end{lemma}

\begin{proof}
From the expressions of the monodromy 
matrices, for all $\vec{t},\vec{v} \in \{0,1\}^{*}_N$,

\begin{align*}(T_N . T'_N)[\vec{t},\vec{v}]  & = 
\sum_{\vec{u}} \text{Tr} (M_N (\vec{t},\vec{u})).
\text{Tr} (M_N (\vec{u},\vec{v})) \\
& = 
\sum_{\vec{u}}
\sum_{\substack{w_1=w_{N+1} \\ w'_1 = w'_{N+1}}}
\sum_{\substack{w_2 \in \{0,1\}\\ w'_2 \in \{0,1\}}} ... 
\sum_{\substack{w_{N} \in \{0,1\} \\ w'_{N} \in \{0,1\}}}  
\prod_{k=1}^{N} R(\vec{t}_k,\vec{u}_k) [w_k,w_{k+1}]
.R'(\vec{u}_k,\vec{v}_k) [w'_k,w'_{k+1}]\\
& = 
\sum_{\substack{w_1=w_{N+1} \\ w'_1 = w'_{N+1}}}
\sum_{\substack{w_2 \in \{0,1\}\\ w'_2 \in \{0,1\}}} ... 
\sum_{\substack{w_{N} \in \{0,1\} \\ w'_{N} \in \{0,1\}}} 
\left( \sum_{\vec{u}}
\prod_{k=1}^{N} R(\vec{t}_k,\vec{u}_k) [w_k,w_{k+1}]
.R'(\vec{u}_k,\vec{v}_k) [w'_k,w'_{k+1}]\right)\\
& = 
\sum_{\substack{w_1=w_{N+1} \\ w'_1 = w'_{N+1}}}
\sum_{\substack{w_2 \in \{0,1\}\\ w'_2 \in \{0,1\}}} ... 
\sum_{\substack{w_{N} \in \{0,1\} \\ w'_{N} \in \{0,1\}}} 
\left( \sum_{\vec{u}}
\prod_{k=1}^{N} R(\vec{t}_k,\vec{u}_k) \otimes R'(\vec{u}_k,\vec{v}_k)
[(w_k,w_{k+1})
, (w'_k,w'_{k+1})]\right)\\
& = Tr (M_N \otimes M'_N(\vec{t},\vec{v})).
\end{align*}
\end{proof}

\begin{figure}[h!]

\[
\begin{tikzpicture}[scale=0.7]
\draw (0,0) grid (1,2);

\draw[line width=0.6mm,color=gray!80] (-0.15,0.5) -- (0.15,0.5);
\draw[line width=0.6mm,color=gray!80] (0.85,0.5) -- (1.15,0.5);
\draw[line width=0.6mm,color=gray!80] (0.5,-0.15) -- (0.5,0.15);
\draw[line width=0.6mm,color=gray!80] (0.5,0.85) -- (0.5,1.15);
\node at (-0.5,0.5) {$w_1$};
\node at (1.5,0.5) {$w'_1$};
\node at (0.5,-0.4) {$t$};
\node at (-2,1) {input};
\node at (3.2,1) {output};

\begin{scope}[yshift=1cm]
\node at (0.5,1.4) {$v$};
\node at (-0.5,0.5) {$w_2$};
\node at (1.5,0.5) {$w'_2$};
\draw[line width=0.6mm,color=gray!80] (-0.15,0.5) -- (0.15,0.5);
\draw[line width=0.6mm,color=gray!80] (0.85,0.5) -- (1.15,0.5);
\draw[line width=0.6mm,color=gray!80] (0.5,-0.15) -- (0.5,0.15);
\draw[line width=0.6mm,color=gray!80] (0.5,0.85) -- (0.5,1.15);
\end{scope}

\draw[thick,decorate,decoration={brace,amplitude=3pt}] (-3.2,1) -- (-3.2,2);

\node at (-4.2,1.5) {$R'(.,v)$};

\draw[thick,decorate,decoration={brace,amplitude=3pt}] (-3.2,0) -- (-3.2,1);

\node at (-4.2,0.5) {$R(t,.)$};

\end{tikzpicture}\]

\caption{\label{figure.product.local.matrices} Graphical representation of the 
matrix $R'  \circ R (t,v)$.}
\end{figure}
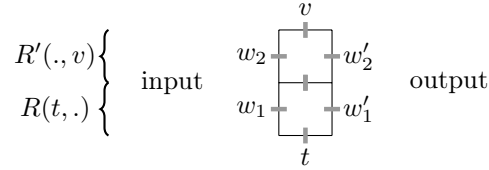

\begin{figure}[h!]

\[
\begin{tikzpicture}[scale=0.7]
\draw (0,0) grid (1,2);

\draw[line width=0.6mm,color=gray!80] (-0.15,0.5) -- (0.15,0.5);
\draw[line width=0.6mm,color=gray!80] (0.85,0.5) -- (1.15,0.5);
\draw[line width=0.6mm,color=gray!80] (0.5,-0.15) -- (0.5,0.15);
\draw[line width=0.6mm,color=gray!80] (0.5,0.85) -- (0.5,1.15);
\node at (0.5,-0.4) {$t$};

\draw[-latex] (1.25,0.5) -- (1.75,0.5); 
\node at (2.5,0) {output};

\begin{scope}[xshift=-2cm,yshift=0.5cm]

\draw (0,0) rectangle (1,1);

\draw[line width=0.6mm,color=gray!80] (-0.15,0.5) -- (0.15,0.5);
\draw[line width=0.6mm,color=gray!80] (0.85,0.5) -- (1.5,0.5) -- (1.5,0) -- (2,0) ;
\draw[line width=0.6mm,color=gray!80] (0.5,-0.15) -- (0.5,0.15);
\draw[line width=0.6mm,color=gray!80] (0.5,0.85) -- (0.5,1.15) -- (0.5,1.5) -- (1.5,1.5) -- (1.5,1) 
-- (2,1);

\draw[-latex] (-0.75,0.5) -- (-0.25,0.5);
\node at (-1.5,0) {input};
\draw[-latex] (0.5,-0.75) -- (0.5,-0.25);

\end{scope}

\begin{scope}[yshift=1cm]
\node at (0.5,1.4) {$v$};
\draw[line width=0.6mm,color=gray!80] (-0.15,0.5) -- (0.15,0.5);
\draw[line width=0.6mm,color=gray!80] (0.85,0.5) -- (1.15,0.5);
\draw[line width=0.6mm,color=gray!80] (0.5,-0.15) -- (0.5,0.15);
\draw[line width=0.6mm,color=gray!80] (0.5,0.85) -- (0.5,1.15);

\draw[-latex] (1.25,0.5) -- (1.75,0.5); 
\end{scope}

\node[scale=1.5] at (4,1) {\textbf{=}}; 

\draw[thick,decorate,decoration={brace,amplitude=3pt}] (0,3.5) -- (1,3.5);

\node at (0.5,4.15) {$R' \circ R (t,v)$};

\draw[thick,decorate,decoration={brace,amplitude=3pt}] (-2,3) -- (-1,3);

\node at (-1.5,3.65) {$Q$};

\begin{scope}[xshift=9cm]

\begin{scope}[xshift=-2cm]

\draw (0,0) grid (1,2);

\draw[line width=0.6mm,color=gray!80] (-0.15,0.5) -- (0.15,0.5);
\draw[line width=0.6mm,color=gray!80] (0.85,0.5) -- (1.5,0.5) -- (1.5,0) -- (2.5,0) -- (2.5,0.5);
\draw[line width=0.6mm,color=gray!80] (0.5,-0.15) -- (0.5,0.15);
\draw[line width=0.6mm,color=gray!80] (0.5,0.85) -- (0.5,1.15);
\node at (0.5,-0.4) {$v$};

\draw[-latex] (-0.75,0.5) -- (-0.25,0.5);

\node at (-1,-0.5) {input}; 

\end{scope}

\begin{scope}[yshift=0.5cm]

\draw (0,0) rectangle (1,1);

\draw[line width=0.6mm,color=gray!80] (-0.15,0.5) -- (0.15,0.5);
\draw[line width=0.6mm,color=gray!80] (0.85,0.5) -- (1.15,0.5) ;
\draw[line width=0.6mm,color=gray!80] (0.5,-0.15) -- (0.5,0.15);
\draw[line width=0.6mm,color=gray!80] (0.5,0.85) -- (0.5,1.15) ;

\draw[-latex] (1.25,0.5) -- (1.75,0.5);
\draw[-latex] (0.5,1.25) -- (0.5,1.75);

\node at (2,-0.5) {output};
\end{scope}

\begin{scope}[yshift=1cm, xshift=-2cm]
\node at (0.5,1.4) {$t$};
\draw[line width=0.6mm,color=gray!80] (-0.15,0.5) -- (0.15,0.5);
\draw[line width=0.6mm,color=gray!80] (0.85,0.5) -- (1.5,0.5) -- (1.5,0) -- (2,0);
\draw[line width=0.6mm,color=gray!80] (0.5,-0.15) -- (0.5,0.15);
\draw[line width=0.6mm,color=gray!80] (0.5,0.85) -- (0.5,1.15);

\draw[-latex] (-0.75,0.5) -- (-0.25,0.5);

\draw[thick,decorate,decoration={brace,amplitude=3pt}] (0,2.5) -- (1,2.5);

\node at (0.5,3.15) {$R \circ R' (v,t)$};

\draw[thick,decorate,decoration={brace,amplitude=3pt}] (2,2) -- (3,2);

\node at (2.5,2.65) {$Q$};

\end{scope}
\end{scope}

\end{tikzpicture}\]

\caption{\label{figure.yang.baxter.equation} Graphical representation of the 
Yang-Baxter equation for the triple $(R,R',Q)$.
In this equation, the input and output of 
the matrix $R' \circ R(t,v).Q$ (represented on the left) and 
$Q.R \circ R'(v,t)$ (represented on the right) are fixed, 
as well as $v,t$.}
\end{figure}
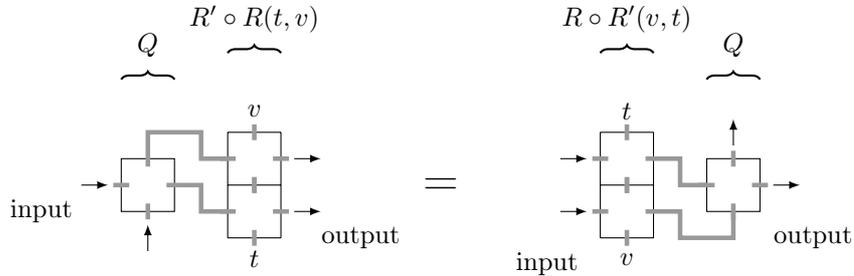

\begin{definition}
Consider $Q$ some north east local function. 
The triple $(R,R',Q)$ verifies the \textbf{Yang-Baxter} equation when for all $t,v \in \{0,1\}$,
the matrices $R' \circ R (t,v)$ 
and $R \circ R' (v,t)$ are conjugated by $Q$: 
\[R' \circ R(t,v).Q = Q.R \circ R'(t,v).\]
This relation is illustrated on Figure~\ref{figure.yang.baxter.equation}.
\end{definition}

\begin{remark}
The Yang-Baxter equation is also often 
called \textbf{star-triangle equation}, due to the 
shapes of diagrams that represent this equation 
under some simplifications (See for instance 
Figure~\ref{figure.simplified.yang.baxter}). 
\end{remark}

\begin{lemma} \label{lemma.commuting.matrices}
If there exists $Q$ an invertible 
north east local matrix function such that 
$(R,R',Q)$ verifies the Yang-Baxter equation, 
then for all $N$: 
\begin{enumerate}
\item The $N$th monodromy matrices $M_N$ 
and $M'_N$ of $R$ and $R'$ verify: 
\[M_N \otimes M'_N = Q M'_N \otimes M_N Q^{-1}.\]
\item 
As a direct consequence the $N$th transfer matrices associated to 
$R$ and $R'$ commute. 
\end{enumerate}
\end{lemma}

\begin{proof}

\begin{enumerate}
\item 
Since $(R,R',Q)$ 
verifies Yang-Baxter equation, for all $\vec{t},\vec{v}$:

\begin{align*}
 M_N \otimes M'_N (\vec{t},\vec{v}) &  = 
\prod_{k=1}^{N}
Q. R \circ R' (\vec{v}_k,\vec{t}_k).Q^{-1}\\
& = Q \left(\prod_{k=1}^{N}
R \circ R' (\vec{v}_k,\vec{t}_k)\right)Q^{-1}\\
& = Q \left( M'_N \otimes M_N (\vec{t},\vec{v})\right) Q^{-1} 
\end{align*}
This can be re-written directly 
\[M_N \otimes M'_N = Q M'_N \otimes M_N Q^{-1}\]
\item 
We have the following equalities, in virtue of 
Lemma~\ref{lemma.product.monodromy.matrices}:
\[(T_N . T'_N)[\vec{t},\vec{v}] = Tr (M_N \otimes M'_N (\vec{t},\vec{v})),\]
\[(T'_N . T_N)_{\vec{t},\vec{v}} = Tr (M'_N \otimes M_N (\vec{t},\vec{v})).\]

We deduce directly that 
\[T_N.T'_N [\vec{t},\vec{v}] = T'_N.T_N [\vec{t},\vec{v}].\]
\end{enumerate}

\end{proof}

\subsection{\label{subsection.yang.baxter.paths} Yang-Baxter paths}

In the following, we will search for paths 
of local matrix function $x \mapsto R^x$ 
such that for all $x,y$, there exists 
some invertible matrix $Q$ such that 
$(R^x,R^y,Q)$ verifies Yang-Baxter equation, 
with an additional constraint of 
rotational symmetry, which means that 
$Q$ corresponds to some $R^z$. 

\begin{definition}
A \textbf{Yang-Baxter path} of local matrix 
functions is an analytic function $x \mapsto R^x$ from $\mathbb{R}$ to $\mathcal{M}_4 (\mathbb{C})$, 
such that there exists some analytic 
function $\delta: \mathbb{R}^2 \mapsto \mathbb{R}$ for which for all $x,y$, 
$(R^x,R^y,Q^{\delta(x,y)})$ verifies Yang-Baxter equation,
where $Q^z$ denotes the north east local 
matrix that corresponds to $R^z$.
\end{definition}

The \textbf{Yang-Baxter path} of $N$th 
transfer matrices corresponding to $x \mapsto R^x$ is 
the analytic function
$x \mapsto T^x_N$, where for all 
$x$, $T^x_N$ is the transfer matrix 
constructed from $R^x$.

For all $u,v,w,w' \in \{0,1\}$,
we denote $a[u,v,w,w'] \in \mathcal{A}$
the symbol in the alphabet of the six 
vertex model whose south (resp. north, west, east) arrow is oriented towards 
north (resp. north, east, west) when $u=0$ (resp. $v=0, w=0, w'=0$), elso towards south (resp. south, west, east).
For instance, $a[0,1,1,0]$ is the symbol: 

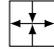
\[\begin{tikzpicture}[scale=0.6]
\draw (0,0) rectangle (1,1);
\draw[-latex] (0.5,0) -- (0.5,0.5);
\draw[-latex] (0.5,1) -- (0.5,0.5);
\draw[-latex] (0.5,0.5) -- (0,0.5);
\draw[-latex] (0.5,0.5) -- (1,0.5);
\end{tikzpicture}\]

\begin{definition}
To a path of local matrix function 
$x \mapsto R^x$ we associate (in 
a bijective way) a \textbf{coefficient attribution function} $\xi: \mathbb{R} \times \mathcal{A} \mapsto 
\mathbb{C}$
such that for all $u,v,w,w' \in \{0,1\}$: 
\[\xi (x,a[u,v,w,w'])= R^x (u,v)[w,w'].\]
The number 
$\xi (x,a)$ for $x \in \mathbb{R}$ and $a \in \mathcal{A}$ is the \textbf{coefficient} of 
$a$ relative to $x$. 
\end{definition}

For a given Yang-Baxter path 
$x \mapsto R^x$, one can represent 
the Yang-Baxter equation for 
$(R^x,R^y,Q^z)$, where $z \equiv \epsilon(x,y)$
in a geometrical way, as on Figure~\ref{figure.simplified.yang.baxter}. 

\begin{figure}[h!]
\[\begin{tikzpicture}[scale=0.7]

\fill[gray!30] (-2.5,2) -- (-2,2.5) -- 
(-1.5,2) -- (-2,1.5) -- (-2.5,2) ; 

\fill[gray!30] (-0.35,2.65) -- (0.35,2.65) -- 
(0.35,3.35) -- (-0.35,3.35) -- (-0.35,2.65) ; 
\draw[dashed] (-0.35,2.65) -- (0.35,2.65) -- 
(0.35,3.35) -- (-0.35,3.35) -- (-0.35,2.65) ; 

\fill[gray!30] (-0.35,0.65) -- (0.35,0.65) -- 
(0.35,1.35) -- (-0.35,1.35) -- (-0.35,0.65) ; 
\draw[dashed] (-0.35,0.65) -- (0.35,0.65) -- 
(0.35,1.35) -- (-0.35,1.35) -- (-0.35,0.65) ;

\draw[line width = 0.5mm] (0,0) -- (0,4);
\draw[line width = 0.5mm] (-1,3) -- (1,3);
\draw[line width = 0.5mm] (-1,1) -- (1,1);
\draw[line width = 0.5mm] (-3,1) -- (-1,3);
\draw[line width = 0.5mm] (-3,3) -- (-1,1);

\node at (-1.25,2) {$x$};
\node at (0.5,3.5) {$z$};
\node at (0.5,1.5) {$y$};

\fill[gray!30] (-6-0.35,-1+2.65) -- (-6+0.35,-1+2.65) -- 
(-6+ 0.35,-1+3.35) -- (-6-0.35,-1+3.35) -- (-6-0.35,-1+2.65) ; 
\draw[dashed] (-6-0.35,-1+2.65) -- (-6+0.35,-1+2.65) -- 
(-6+ 0.35,-1+3.35) -- (-6-0.35,-1+3.35) -- (-6-0.35,-1+2.65) ; 
\draw[line width=0.5mm] (-6.75,2) -- (-5.25,2);
\draw[line width=0.5mm] (-6,2.75) -- (-6,1.25);
\node at (-5.25,2.75) {$x$};
\node[scale=0.75] at (-6,1) {south};
\node[scale=0.75] at (-7.5,2) {west};

\draw[-latex,dashed] (-3,2) -- (-5,2);

\draw[dashed] (-2.5,2) -- (-2,2.5) -- 
(-1.5,2) -- (-2,1.5) -- (-2.5,2) ; 

\node at (0,-0.5) {$t$};
\node at (0,4.5) {$v$};
\node at (-3.5,1) {$w_1$};
\node at (-3.5,3) {$w_2$};

\node at (1.5,1) {$w'_1$};
\node at (1.5,3) {$w'_2$};

\node[scale=2] at (3,2) {$=$};

\begin{scope}[xshift=6cm]

\fill[gray!30] (1.5,2) -- (2,2.5) -- 
(2.5,2) -- (2,1.5) -- (1.5,2) ; 

\draw[dashed] (1.5,2) -- (2,2.5) -- 
(2.5,2) -- (2,1.5) -- (1.5,2) ;

\fill[gray!30] (-0.35,2.65) -- (0.35,2.65) -- 
(0.35,3.35) -- (-0.35,3.35) -- (-0.35,2.65) ; 
\draw[dashed] (-0.35,2.65) -- (0.35,2.65) -- 
(0.35,3.35) -- (-0.35,3.35) -- (-0.35,2.65) ; 

\fill[gray!30] (-0.35,0.65) -- (0.35,0.65) -- 
(0.35,1.35) -- (-0.35,1.35) -- (-0.35,0.65) ; 
\draw[dashed] (-0.35,0.65) -- (0.35,0.65) -- 
(0.35,1.35) -- (-0.35,1.35) -- (-0.35,0.65) ;

\node at (-1.5,1) {$w_1$};
\node at (-1.5,3) {$w_2$};

\node at (3.5,1) {$w'_1$};
\node at (3.5,3) {$w'_2$};

\node at (2.75,2) {$x$};
\node at (0.5,3.5) {$z$};
\node at (0.5,1.5) {$y$};

\draw[line width = 0.5mm] (0,0) -- (0,4);
\draw[line width = 0.5mm] (-1,3) -- (1,3);
\draw[line width = 0.5mm] (-1,1) -- (1,1);
\draw[line width = 0.5mm] (1,1) -- (3,3);
\draw[line width = 0.5mm] (1,3) -- (3,1);

\node at (0,-0.5) {$t$};
\node at (0,4.5) {$v$};

\end{scope}

\end{tikzpicture}\]
\caption{\label{figure.simplified.yang.baxter} 
A geometrical representation 
of Yang-Baxter equation when $x \mapsto R^x$ is 
a Yang-Baxter path. The two diagrams 
represent a sum of complex numbers. 
The one on the left represents the sum on $a_1,a_2,a_3$  
of the numbers $\xi(x,a_1) \xi (y,a_2) \xi (z,a_3)$ such that there exists an orientation 
of the diagram such that there is an arrow 
entering in an edge if and only if the corresponding 
letter is $0$ (else the arrow is outgoing), and one can observe in 
the gray square corresponding to $x,y,z$ 
(oriented according to these respective 
letters) the 
respective symbols $a_1,a_2,a_3$. An 
exemple of orientation for given letters 
$u,v,w_1,w_2,w'_1,w'_2$ is given on Figure~\ref{figure.example.yang.baxter.simplified}. 
The diagram on the right represents a sum 
defined in a similar way. 
}
\end{figure}

\begin{figure}[h!]

\[\begin{tikzpicture}[scale=0.5]

\draw[line width = 0.5mm] (0,0) -- (0,4);
\draw[line width = 0.5mm] (-1,3) -- (1,3);
\draw[line width = 0.5mm] (-1,1) -- (1,1);
\draw[line width = 0.5mm] (-3,1) -- (-1,3);
\draw[line width = 0.5mm] (-3,3) -- (-1,1);

\draw[line width = 0.4mm,-latex] (-2,2) -- (-2.75,2.75);
\draw[line width = 0.4mm,-latex] (-3,1) -- (-2.25,1.75);

\draw[line width = 0.4mm,-latex] (0,3) -- (0,3.75);

\draw[line width = 0.4mm,-latex] (0,0) -- (0,0.75);

\draw[line width = 0.4mm,-latex] (-1,1) -- (-0.25,1);

\draw[line width = 0.4mm,-latex] (0,3) -- (-0.75,3);

\draw[line width = 0.4mm,-latex] (0,1) -- (0,2.5);

\draw[line width = 0.4mm,-latex] (0,1) -- (0.75,1);

\draw[line width = 0.4mm,-latex] (1,3) -- (0.25,3);

\node at (0,-0.5) {$0$};
\node at (0,4.5) {$1$};
\node at (-3.5,1) {$0$};
\node at (-3.5,3) {$1$};

\node at (1.5,1) {$1$};
\node at (1.5,3) {$0$};
\end{tikzpicture}\]

\caption{\label{figure.example.yang.baxter.simplified}Example of orientation of the left diagram in 
Figure~\ref{figure.simplified.yang.baxter} according to fixed 
border letters.}
\end{figure}

\subsection{\label{section.strongly.symmetric} Strongly symmetric solutions of Yang-Baxter 
equation}

In this section, we present a criterion 
based on a symmetry condition 
on $R,R'$ that allows to ensure the existence 
of some $Q$ (not necessarily invertible) 
such that $(R,R',Q)$ verifies 
Yang-Baxter equation. Implicitely, this criterion 
relies also on the symmetries verified by 
the model.

\begin{definition}
Let $R$ be a local matrix function represented 
as in Fact~\ref{fact.local.matrix}. This 
function is said to be \textbf{strongly 
symmetric} when $e=b$, $f=a$, $c=d$ and 
$ab\neq 0$. In this case, we denote 
\[\delta(R) = \frac{a^2+b^2-c^2}{2ab}.\] 
Let $R',R''$ be two other strongly symmetric 
local matrix functions 
obtained respectively by replacing 
$a,b,c$ by $a',b',c'$ and $a'',b'',c''$. 
We denote $(S)_{R,R',R''}$ the system of equations: 
\[(S)_{R,R',R''}: \left\{\begin{array}{rcl}
a'ca'' & = & b'cb'' + c'ac''\\
a'bc'' & = & b'ac'' + cc'b''\\
c'ba'' & = & c'ab''+ b'cc''
\end{array}\right.\]

\end{definition}

\begin{lemma}
Let $R$,$R'$ and $R''$ three strongly 
symmetric local matrix functions and $Q$ 
the north east local matrix function corresponding 
to $R''$. Then $(R,R',Q)$ verifies Yang-Baxter 
equation if and only if $(S)_{R,R',R''}$ 
and $(S)_{R',R,R''}$ are verified.
\end{lemma}

 \begin{proof}

\begin{itemize}

\item \textbf{Expression of the matrices 
involved:}

Let us denote
\[R= \left( \begin{array}{cccc}
a & 0 & 0 & 0 \\
0 & b & c & 0 \\
0 & c & b & 0 \\
0 & 0 & 0 & a \\
\end{array} \right), \quad R'= \left( \begin{array}{cccc}
a' & 0 & 0 & 0 \\
0 & b' & c' & 0 \\
0 & c' & b' & 0 \\
0 & 0 & 0 & a' \\
\end{array} \right).\]

Let us consider $Q$ a north east local 
matrix function represented as: 
\[Q = \left( \begin{array}{cccc}
a'' & 0 & 0 & 0 \\
0 & c'' & b'' & 0 \\
0 & b'' & c'' & 0 \\
0 & 0 & 0 & a'' \\
\end{array} \right),\]

One can represent $R' \circ R$ by 
a matrix in $\mathcal{M}_8 (\mathbb{C})$: 

\[R' \circ R = 
\left( \begin{array}{cc} R(0,0) 
\otimes R' (0,0) + R(0,1) 
\otimes R' (1,0) & R(0,0) 
\otimes R' (0,1) + R(0,1) 
\otimes R' (1,1)\\
R(1,0) 
\otimes R' (0,0) + R(1,1) 
\otimes R' (1,0) & R(1,0) 
\otimes R' (0,1) + R(1,1) 
\otimes R' (1,1)\\\end{array} 
\right).\] 

After direct computation:

\[\begin{tikzpicture}
\node at (0,0) {$
R' \circ R = \left( \begin{array}{cccccccc} aa' & 0 & 0 & 0 & 0 & 0 & 0 & 0  \\
 0 & a b' & 0 & 0 & a c' & 0 & 0 & 0 \\
 0 & c c' &  ba' & 0 & cb' & 0 & 0 & 0 \\
  0 & 0 & 0 & b b' & 0 & ca' & b c' & 0 \\
  
  0  & bc' & c a' & 0 & b b'& 0& 0 & 0  \\
   0 & 0 & 0 & c b' & 0 & b a' & c c'& 0 \\
     0 & 0 & 0 & a c'& 0 & 0 & a b'& 0  \\
       0 & 0 & 0 & 0 & 0 & 0 & 0 & aa' \\
\end{array} 
\right)$}; 
\draw[dashed,color=gray!95] (-2.4,0.025) -- (3.8,0.025); 
\draw[dashed,color=gray!95] (0.7,1.8) -- (0.7,-1.8);
\end{tikzpicture}.\] 
We obtain $R \circ R'$ by exchanging $a$ (resp. 
$b$,$c$) with $a'$ (resp. $b'$, $c'$).

\item \textbf{Reduction of Yang-Baxter 
equation to $(S)_{R,R',R''}$:}

The Yang-Baxter equation on $(R,R',Q)$ tells 
that $\overline{Q}.R'\circ R = R\circ R' \overline{Q}$.

When writing down the equations obtained 
writing Yang-Baxter equation block by 
block, we obtain a lot of trivial equations. 

For instance, the north west block equation 
is written: 
\[\left( \begin{array}{cccc}
a'' & 0 & 0 & 0 \\
0 & c'' & b'' & 0 \\
0 & b'' & c'' & 0 \\
0 & 0 & 0 & a''
\end{array} \right). 
\left( \begin{array}{cccc}
aa' & 0 & 0 & 0 \\
0 & ab' & 0 & 0 \\
0 & cc' & ba' & 0 \\
0 & 0 & 0 & bb'
\end{array} \right) = \left( \begin{array}{cccc}
aa' & 0 & 0 & 0 \\
0 & a'b & 0 & 0 \\
0 & cc' & b'a & 0 \\
0 & 0 & 0 & bb'
\end{array} \right). \left( \begin{array}{cccc}
a'' & 0 & 0 & 0 \\
0 & c'' & b'' & 0 \\
0 & b'' & c'' & 0 \\
0 & 0 & 0 & a''
\end{array} \right)
 \]
 which is equivalent to:
 \[\left( \begin{array}{cccc}
aa'a'' & 0 & 0 & 0 \\
0 & ab'c''+cc'b'' & ba'b'' & 0 \\
0 & b''ab'+c''cc' & c''ba' & 0 \\
0 & 0 & 0 & bb'a''
\end{array} \right) = \left( 
\begin{array}{cccc}
aa'a'' & 0 & 0 & 0 \\
0 & a'bc'' & a'bb'' & 0 \\
0 & b''ab'+c''cc' & cc'b'' + b'ac'' & 0 \\
0 & 0 & 0 & bb'a''
\end{array} \right),\]
which reduces to the equation $ba'c'' = 
ab'c'' + cc'b''$.

After simplification of the other equations, the Yang-Baxter equation reduces 
to: 

\[(S)_{R,R',R''}: \left\{\begin{array}{rcl}
a'ca'' & = & b'cb'' + c'ac''\\
a'bc'' & = & b'ac'' + cc'b''\\
c'ba'' & = & c'ab''+ b'cc''
\end{array}\right.\]
\end{itemize}

\end{proof}

\begin{lemma}
Let $R$ and $R'$ be two strongly symmetric 
local matrix functions such that $\delta(R)=\delta(R')$. There 
exists $Q$ associated to a strongly symmetric local 
matrix function such that $(R,R',Q)$ 
verifies the Yang-Baxter equation.
\end{lemma}

\begin{proof}
Let $Q$ be a north 
east local matrix function and $R''$ the local matrix function corresponding 
to $Q$. The system $(S)_{R,R',R''}$ 
has a solution in $a'',b'',c''$ 
if and only if the determinant  
\[\left|\begin{array}{ccc} -a'c & b'c & c'a \\0 & cc' & (b'a-ba') \\
-c'b & c'a & b'c \end{array} \right|= 
- a'c( b'c' c^2 - c'a(b'a-ba')) - c'b (b'c(b'a-ba') - c'^2 a c) \]
is zero, meaning 
\[ a'b' c' c^3 - a'b'c'c a^2 + abcc'a'^2 + abcc'b'^2 - a'b'c'cb^2 - abcc'^3 = 0\]
As a consequence, factoring by $2aa'bb'cc'$, 
this is equivalent to 
\[cc' (\delta(R')-\delta(R))= 0,\]
which is true by hypothesis. 
As a consequence, the system $(S)_{R,R',R''}$ has a solution.
\end{proof}

\section{\label{section.algebraic.bethe.ansatz} Algebraic Bethe ansatz}

Let us consider some analytic function $x \mapsto R^x$ 
from $\mathbb{C}$ to $\mathcal{M}_4 (\mathbb{C})$
whose images are strongly symmetric local matrix 
functions
 and for all $x \in \mathbb{C}$, denote:
\[R^x = \left( \begin{array}{cccc}
a(x) & 0 & 0 & 0 \\
0 & b(x) & c(x) & 0 \\
0 & c(x) & b(x) & 0 \\
0 & 0 & 0 & a(x) 
\end{array}\right).\]
Let us denote $Q^x$ the north east local 
matrix function associated to $R^x$: 

\[Q^x = \left( \begin{array}{cccc}
a(x) & 0 & 0 & 0 \\
0 & c(x) & b(x) & 0 \\
0 & b(x) & c(x) & 0 \\
0 & 0 & 0 & a(x) 
\end{array}\right).\]

We assume that: 

\begin{enumerate}
 
\item there exists $\delta : \mathbb{C}^2 \mapsto \mathbb{C}$ such that for all $x,y$, 
$(R^x,R^y,Q^{\delta(x,y)})$ verifies Yang-Baxter 
equation 
\item $b(\delta(x,y)) = 0 \Leftrightarrow b(\delta(y,x))=0.$
\item for all $x,y,z \in \mathbb{C}$, $  
\delta(\delta(x,y),
\delta(x,z)) = \delta(y,z)$.
\end{enumerate}

\noindent Let us denote $\mathcal{D}$ 
the set $(x,y) \in \mathbb{C}^2$ such that 
$b(\delta(x,y)) \neq 0$. 
We also 
assume that for all $(x,y) \in \mathcal{D}$,
\[\frac{c(\delta(x,y))}{b(\delta(x,y))}
= - \frac{c(\delta(y,x))}{b(\delta(y,x))}.\]

We describe in 
this section an ansatz (meaning a candidate eigenvector) for the eigenequation 
of the transfer matrix associated to any 
$R^x$. In this version of the Algebraic Bethe ansatz, we extracted the conditions 
on the parameters of the path of local matrix 
functions allowing the ansatz, in order to 
understand how the ansatz works in 
general, and in the perspective of 
extending it to other models. 
In Section~\ref{section.commutation.relations.algebraic.ansatz}, 
we derive some commutation relations on 
the components of the monodromy matrices 
that come from Yang-Baxter equation.
In Section~\ref{section.algebraic.ansatz}, 
we state the Algebraic Bethe ansatz for 
the path $x \mapsto R^x$ using these commutation 
relations. In the end, we apply this method 
to a particular \textbf{trigonometric} 
Yang-Baxter path [Section~\ref{section.from.algebraic.to.coordinate}]. \bigskip

Let us recall that for all 
$x,y$, $z=\delta(x,y)$, 
Yang-Baxter equation is reduced to 

\[\left\{\begin{array}{rcl}
a(x)c(y)a(z) & = & b(x)c(y)b(z) + c(x)a(y)c(z)\\
a(x)b(y)c(z) & = & b(x)a(y)c(z) + c(x)c(y)b(z)\\
c(x)b(y)a(z) & = & c(x)a(y)b(z)+ b(x)c(y)c(z)
\end{array}\right.\]

\subsection{\label{section.commutation.relations.algebraic.ansatz} Commutation of Lax matrices components} 

\begin{notation}
We denote $M^x_N$ the $N$th monodromy 
matrix associated to $R^x$, represented as:
\[M^x_{N} = \left( \begin{array}{cc} 
M^x_{N} (0,0) & M^x_{N} (0,1) \\
M^x_{N} (1,0) & M^x_{N} (1,1)
\end{array}\right) = \left( \begin{array}{cc} 
A_N (x) & C_N (x) \\
B_N (x) & D_N (x) 
\end{array}\right).\]
\end{notation}

\begin{notation}
Let us also denote for all 
$(x,y) \in \mathcal{D}$:
\[\lambda(x,y) = \frac{a(\delta(x,y))}{b(\delta(x,y))} \qquad \text{and} \qquad \mu(x,y) = -\frac{c(\delta(x,y))}{b(\delta(x,y))}.\]
\end{notation}

\begin{fact}
The function $\mu : (x,y) \mapsto \mu(x,y)$ is antisymetric: 
for all $x,y$, $\mu(x,y)=-\mu(y,x)$, by hypothesis on 
the antisymmetry of $(c/b) \circ \delta$.
\end{fact}

\begin{lemma}
\label{lemma.commutation.relations.monodromy}
For all $x,y$ 
such that $b(\delta(x,y)) \neq 0$:
\[\left\{ \begin{array}{rcl}
A_N (x) B_N (y
) & = &  \lambda(x,y) B_N (y) 
A_N (x) + \mu(x,y) B_N (x) A_N (y) \\
D_N (x) B_N (y
) & = & \lambda(y,x) B_N (y) 
D_N (x) +\mu(y,x) B_N (x) D_N (y)  \\
B_N (x) A_N (y) & = & \lambda(x,y) A_N (y) B_N(x) +\mu(x,y) A_N (x) B_N (y)\\
\end{array}\right.\]
Moreover, for all $x,y$ such that $a(\delta(x,y)) \neq 0$: 
\[\left\{\begin{array}{rcl} 
B_N (x) B_N (y) & = & B_N (y) B_N (x)\\
A_N (x) A_N (x) & = & A_N (y) A_N (x).
\end{array}\right.\]
\end{lemma}

\begin{proof}
\begin{enumerate}
\item \textbf{Equation on the monodromy matrices:} 

We know by Lemma~\label{lemma.commuting.matrices} for all $x,y$,
\[M^x_{N} \otimes M^y_{N}.Q^{\delta(x,y)} = 
Q^{\delta(x,y)}. M^y_{N} \otimes M^x_{N}.\]
For all $x,y$ the matrix $M^y_{N} \otimes M^x_{N}$
is:
\[ \begin{tikzpicture} 
\node at 
(0,0) {$\left( \begin{array}{cccc} 
A_N (y) A_N (x) & A_N (y) C_N (x) & C_N (y) A_N (x) & C_N (y) C_N (x) \\
A_N (y) B_N (x) & A_N (y) D_N (x) & C_N (y) B_N (x) & C_N (y) D_N (x) \\
B_N (y) A_N (x) & B_N (y) C_N (x) & D_N (y) A_N (x) & D_N (y) C_N (x) \\
B_N (y) B_N (x) & B_N (y) D_N (x) & D_N (y) B_N (x) & D_N (y) D_N (x) \\
\end{array}\right)$};
\node[gray!99] at (-3.75,1.25) {(0,0)};
\node[gray!99] at (-1.25,1.25) {(0,1)};
\node[gray!99] at (1,1.25) {(1,0)};
\node[gray!99] at (3.5,1.25) {(1,1)};
\node[gray!99] at (5.5,0.625) {(0,0)};
\node[gray!99] at (5.5,0.185) {(0,1)};
\node[gray!99] at (5.5,-0.24) {(1,0)};
\node[gray!99] at (5.5,-0.64) {(1,1)};
\end{tikzpicture}\]

\item \textbf{On coefficients $((1,0),(0,0))$:}

In particular, from the equality on the 
coefficient $((1,0),(0,0))$ of this last 
equality:
\[a(\delta(x,y)). B_N (y) 
A_N (x) = b(\delta(x,y)). A_N (x) B_N (y
) + c (\delta(x,y)). B_N (x) A_N (y)\]

By dividing by $b(\delta(x,y))$, 
we get the first equation stated 
in the lemma.

\item \textbf{Coefficients $((1,1),(0,1))$:} we obtain:

\[c(\delta(x,y)). B_N (y) D_N (x) + b(\delta(x,y)). D_N (y) B_N (x) = a(\delta(x,y)) . B_N (x) D_N (y).\]
After exchanging $x$ and $y$: 

\[ b(\delta(y,x)). D_N (x) B_N (y) = a(\delta(y,x)) . B_N (y) D_N (x) - c(\delta(y,x)). B_N (x) D_N (y) .\]

We obtain the second equation stated in the lemma
by dividing by $b(\delta(x,y))$.
\item The third equality is 
obtained from coefficients $((0,1),(0,0))$. 
The second system is obtained from 
coefficients $((1,1),(0,0))$ and $((0,0),(0,0))$.
\end{enumerate}

\end{proof}

As a consequence, we have the following statement:

\begin{lemma} \label{lemma.commutation.relations} 
For all $N$, $1 \le n \le N$, $x \in \mathbb{C}$ and a finite sequence $(x_1, ... , x_n) \in 
\mathbb{C}^n$ of distinct numbers and distinct 
from $x$:
\[\begin{array}{rcl} 
\displaystyle{A_N (x).\left(\prod_{k=1}^n B_N(x_k)\right)} & = &
\displaystyle{\left( \prod_{k=1}^n \lambda (x,x_k) B_N(x_k) \right). A_N (x)} \\
& & + \displaystyle{
B_N(x)\left( \sum_{j=1}^n \mu (x,x_j).
\left( \prod_{k\neq j} \lambda (x_j,x_k) B_N(x_k)\right) A_N (x_j)\right)} \end{array}\]
A similar equality is verified in which $A$
is replaced by $D$ and in 
each copy of $\lambda$ and $\mu$ the 
two arguments are exchanged. Another 
similar equality is obtained by exchanging only
$B$ and $A$.
\end{lemma}

\begin{proof}

\begin{itemize}
\item \textbf{Auxiliary equality derived from 
Yang-Baxter equation:} 

For all $x,y,z$ all distinct:
\[ \mu (x,z) \lambda (z,y) = \lambda (x,y) \mu (x,z)
+ \mu (x,y) \mu (y,z).\]

This equation derives from 
the third equation 
of the system $(S)_{R,R',R''}$, 
where $(R,R',R'')$ is equal to $(R^{\delta(z,x)},R^{\delta(z,y)},
R^{\delta(\delta(z,x),\delta(z,y))})= 
(R^{\delta(z,x)},R^{\delta(z,y)},
R^{\delta(x,y)})$:

\[\begin{array}{rcl} c(\delta(z,x)) b(\delta(z,y)) a(\delta(x,y)) & = & c(\delta(z,x)) a(\delta(z,y)) b(\delta(x,y)) \\
& & + 
c(\delta(z,y)) c(\delta(x,y)) b(\delta(z,x))
\end{array}.\]

After dividing by $b(\delta(z,y))b(\delta(z,x))
b(\delta(\delta(z,x),\delta(z,y)))$:

\[\frac{a(\delta(x,y))}{b(\delta(x,y))} \frac{c(\delta(z,x))}{b(\delta(z,x))}
= \frac{c(\delta(z,x))}{b(\delta(z,x))} \frac{a(\delta(z,y))}{b(\delta(z,y))} 
+ \frac{c(\delta(z,y))}{b(\delta(z,y))} \frac{c(\delta(x,y))}{b(\delta(x,y))}\]

This can be rewritten:

\[\lambda (x,y) \mu (z,x)
=  \mu (z,x) \lambda(z,y) 
 - \mu(x,y) \mu (z,y).\]
We then use the antisymetry of $\mu$ and then multiply by $-1$ to get: 

\[\lambda (z,y)\mu (x,z) 
= \mu (x,z) \lambda(x,y) 
+ \mu (x,y) \mu (y,z).\]

\item \textbf{Proof of the statement 
by recursion:}

The statement of the lemma for $n=1$ 
derives directly from 
Lemma~\ref{lemma.commutation.relations.monodromy}.  Let us assume that it 
is verified for some $n \ge 1$,
and consider a real number $x$ 
and a sequence $(x_1, ... , x_{n+1})$.

\item \textbf{Application of recursion 
hypothesis on $1,x,(x_1)$:}

We first apply the equality for 
$x,(x_{1})$:
\[\begin{array}{rcl} 
\displaystyle{A_N (x).\left(\prod_{k=1}^{n+1} B_N(x_k)\right)} & = &
\displaystyle{\lambda (x,x_{1}) B_N(x_{1}). A_N (x). \left(\prod_{k=2}^{n+1} B_N(x_k)\right)} \\
& & + \displaystyle{
B_N(x) \mu (x,x_1).
 A_N (x_1). \left(\prod_{k=2}^{n+1} B_N(x_k)\right)} \end{array}\]
 
\item \textbf{Application of recursion 
hypothesis on $n$:}

We apply then the equality on 
$n,x,(x_2,...,x_{n+1})$ and 
$n, x_1 , (x_2, ... , x_{n+1})$ and obtain: 

\[\begin{array}{c}
\displaystyle{A_N (x).\left(\prod_{k=1}^{n+1} B_N(x_k)\right) = \lambda (x,x_{1}) B_N(x_{1}). \left( \prod_{k=2}^{n+1} \lambda (x,x_k) B_N(x_k) \right). A_N (x)}  \\ 
+ \displaystyle{\lambda (x,x_{1}) B_N(x_{1}). B_N (x) \left(
\sum_{j=2}^{n+1} \mu (x,x_j). 
\left( \prod_{\substack{k=2 \\ k \neq j}}^{n+1} \lambda (x_j,x_k) B_N(x_k)\right) A_N (x_j)\right)}
\\
 + \displaystyle{B_N (x) \mu (x,x_1).\left( \prod_{k=2}^{n+1} \lambda (x,x_k) B_N(x_k) \right). A_N (x_1)}\\
 + \displaystyle{B_N (x) \mu (x,x_1) 
 B_N (x_1). \left(
\sum_{j=2}^{n+1} \mu (x_1,x_j). 
\left( \prod_{\substack{k=2 \\ k \neq j}}^{n+1} \lambda (x_j,x_k) B_N(x_k)\right) A_N (x_j)\right)}\end{array} \]

\item \textbf{Application of the auxiliary 
equation:}
 
We group the second and last terms of 
the second member of this equality, and 
use the auxillary equality 
replacing $x,y,z$ by $x,x_1,x_j$: 

\[\lambda(x,x_1) \mu (x,x_j)
+ \mu (x,x_1) \mu (x_1,x_j) 
= \mu (x,x_j) \lambda (x_j,x_1)\]

Thus we have: 

\[\begin{array}{c}
\displaystyle{A_N (x).\left(\prod_{k=1}^{n+1} B_N(x_k)\right) = \lambda (x,x_{1}) B_N(x_{1}). \left( \prod_{k=2}^{n+1} \lambda (x,x_k) B_N(x_k) \right). A_N (x)}  \\ 
+ \displaystyle{ B_N(x_{1}). B_N (x) \left(
\sum_{j=2}^{n+1}\mu (x,x_j) \lambda(x,x_1) . 
\left( \prod_{\substack{k=2 \\ k \neq j}}^{n+1} \lambda (x,x_k) B_N(x_k)\right) A_N (x_j)\right)}
\\
 + \displaystyle{B_N (x) \mu (x,x_1).\left( \prod_{k=2}^{n+1} \lambda (x,x_k) B_N(x_k) \right). A_N (x_1)}\\
 \end{array} \]
 We then group the two last terms 
 in order to obtain the formula.
\end{itemize}
\end{proof}

\subsection{\label{section.algebraic.ansatz} Statement of the ansatz} 

The algebraic Bethe ansatz consists in
proposing eigenvectors in the 
space $\Omega_N$ for a
transfer matrix $T_N (x)$, $x \in \mathbb{C}$  , 
using the commutation relations on 
the components of the monodromy matrices 
proved in Section~\ref{section.commutation.relations.algebraic.ansatz}. 

\begin{notation}
In the following, for all $N \ge 1$, we denote 
$\boldsymbol{\nu}_N = \ket{0 ... 0}$ 
in the canonical basis of $\Omega_N$. This vector 
is often called \textbf{vacuum state} (since it is empty 
from $1$ symbols).
\end{notation}

Let us fix some $x \in \mathbb{C}$, some 
integer $n \le N$ and a sequence $\textbf{x}=(x_1,...,x_n) \in \mathbb{C}^n$ such that none 
of the numbers $b(\delta(x_j,x_k))$, 
$b(\delta(x,x_j))$, $j,k \le n$ is 
zero. 

\begin{notation} Let us denote $\psi_{n,N} (\textbf{x})$ 
the vector:
\[\psi_{n,N} (\textbf{x}) = (B_N (x_1)... B_N (x_n)).{\boldsymbol{\nu}}_N.\]
\end{notation}

\begin{theorem}[Algebraic Bethe ansatz]
\label{theorem.algebraic.bethe.ansatz}
Let us assume that for all $j$:
\[a(x_j)^N \prod_{\substack{k=1 \\ k \neq j}}^{n} \frac{a(\delta(x_j,x_k))}{b(\delta(x_j,x_k))}  = 
 b(x_j)^N \prod_{\substack{k=1 \\ k \neq j}}^{n}
\frac{a(\delta(x_k,x_j))}{b(\delta(x_k,x_j))}\]
Then 
$T_N (x) . \psi_{n,N} (\textbf{x}) = 
\Lambda_{x,n,N} (x_1,...,x_n) \psi_{n,N} (\textbf{x})$,
where 
\[\Lambda_{x,n,N} (x_1,...,x_n) = a(x)^N 
\prod_{k=1}^n \frac{a(\delta(x,x_k))}{b(\delta(x,x_k))}
+ b(x)^N 
\prod_{k=1}^n \frac{a(\delta(x_k,x))}{b(\delta(x_k,x))}.\]
\end{theorem}

\begin{remark}
The vector $\psi_{n,N} (\textbf{x})$ 
is a candidate eigenvector, since we don't know 
if this vector is zero or not.
\end{remark}

\begin{proof}
We have, by definition of the transfer matrix, 
that: 
\[T_N (x) \cdot \psi_{n,N} (\textbf{x}) = (A_N (x) + D_N (x)) \cdot \psi_{n,N} (\textbf{x})\]
By Lemma~\ref{lemma.commutation.relations}, 
and since for all $z$, 
$A_N (z) \cdot \boldsymbol{\nu}_N = a(z)^N.\boldsymbol{\nu}_N$ and $D_N (z) \cdot \boldsymbol{\nu}_N = b(z)^N \cdot \boldsymbol{\nu}_N$,

\[A_N (x) \cdot \psi_{n,N} (\textbf{x}) = a(x)^N \prod_{k=1}^n \lambda (x,x_k) \psi_{n,N} (\textbf{x}) + \left( \sum_{j}
a(x_j)^N \mu(x,x_j) \prod_{k\neq j} \lambda(x_j,x_k)\right) \cdot B_N (x) \cdot \prod_{k\neq j} 
B_N (x_k).\nu_N.
\]
Similarly we have: 
\[D_N (x) \cdot \psi_{n,N} (\textbf{x})  = b(x)^N \prod_{k=1}^n \lambda (x_k,x) \psi_{n,N} (\textbf{x}) + \left( \sum_{j}
b(x_j)^N \mu(x_j,x) \prod_{k\neq j} \lambda(x_k,x_j)\right) \cdot B_N (x) \cdot \prod_{k\neq j} 
B_N (x_k) \cdot \nu_N.
\]

\noindent Since $\mu(x,x_j) = - \mu(x_j,x)$ (by 
antisymetry of $(c/b)\circ \delta$), and 
from the definition of $\lambda$, we obtain 
that 

\[T_N (x) \cdot \psi_{n,N} (\textbf{x}) = \Lambda_{x,n,N} (x_1,...,x_n) 
\cdot \psi_{n,N} (\textbf{x}).\]
\end{proof}

\subsection{Expression of the candidate eigenvector}

This part is the development of an idea that 
one can find in~\cite{Korepin} (Appendix VII.2). 
The formula (A.2.4) that according to the authors, 
helped to derive 
the result of the coordinate Bethe ansatz 
from the algebraic one, was proved in this text 
for $N=2$ (Formula 5.9).

\begin{notation}
For all $n$ integer, $\textbf{x}=(x_1,...,x_n)$ 
a sequence of complex numbers, and $j \le n$, 
we denote $\textbf{x}_{*}^j$ the 
sequence obtained by suppressing the $j$th 
element: 
\[\textbf{x}^j_{*}= (x_1,...,x_{j-1},x_{j+1},...,x_n).\] 
For a sequence $\boldsymbol{\epsilon}$ 
in $\{0,1\}^N$, we denote 
\[\boldsymbol{\epsilon}_{*} = (\boldsymbol{\epsilon}_1,...,\boldsymbol{\epsilon}_{N-1}).\]
\end{notation}

\begin{lemma}
\label{lemma.recursion.monodromy.candidate}
Let $N \ge 1$ be an integer, $1 \le n \le N$ and 
$(x_1,...,x_n)$ a sequence of complex 
numbers with $n \le N$.
We have the following equality: 
\begin{align*}
\prod_{k=1}^n B_{N+1} (x_k) & = \left( \prod_{k=1}^n B_N (x_k) \right) \otimes \left( \prod_{k=1}^n A_1 (x_k)\right) \\
& \quad + \sum_{j=1}^n \prod_{k \neq j} \lambda(x_k,x_j).
\left( \prod_{k \neq j} B_N (x_k) D_N (x_j)\right)\otimes \left( B_1 (x_j) \prod_{k \neq j} A_1 (x_k) \right) \\
\end{align*}
\end{lemma}

\begin{remark}
Let us notice that the proof, although 
simply relying on commutation relations, 
is dependant upon the direction of the 
product made in the recursion argument.
\end{remark}

\begin{proof}
This statement is true for all $N$ and $(x_1,...,x_n)$ when $n=1$: 
this comes from the fact that for all $x$,
\[B_{N+1} (x) = B_N (x) \otimes A_1 (x) 
+ D_N (x) \otimes B_1 (x).\]
Let us assume the statement for some $n$, 
and prove it for $n+1$. Let us consider some sequence $(x_1,...,x_{n+1})$. Using 
the hypothesis for the sequence $(x_2,...,x_{n+1})$, the expression 
\[B_{N+1} (x_1) = B_N (x_1) \otimes A_1 (x_1) 
+ D_N (x_1) \otimes B_1 (x_1),\]
and the fact that for all $x,y$, $B_1 (x) B_1 (y) = 0$,
we get the following equality: 

\begin{align*}
\prod_{k=1}^{n+1} B_{N+1} (x_k) & = 
\left( \prod_{k=1}^{n+1} B_N (x_k) \right) \otimes \left( \prod_{k=1}^{n+1} A_1 (x_k) \right)\\
& \quad + \sum_{j=2}^{n+1}
\prod_{\underset{k \neq j}{k \ge 2}} \lambda(x_k,x_j). \left( \prod_{\underset{k \neq j}{k \ge 1}}B_N (x_k) D_N (x_j) \right) \otimes \left( A_1 (x_1) B_1 (x_j) \prod_{\underset{k \neq j}{k \ge 2}} A_1 (x_k)\right) \\
& \quad + \left( D_N (x_1) \prod_{k=2}^{n+1} B_N (x_k) \right) \otimes \left( B_1 (x_1) \prod_{k=2}^{n+1}A_1 (x_k) \right)
\end{align*}

We know that: 

\begin{align*}
D_N (x_1) \prod_{k=2}^{n+1} B_N (x_k) & = 
\prod_{k=2}^{n+1} \lambda(x_k,x_1) B_N (x_k) D_N (x_1) \\
& \quad + B_N (x_1) \left( \sum_{j=2}^{n+1}
\mu(x_j,x_1). \left( \prod_{\underset{k \ge 2}{k \neq j}} \lambda (x_k,x_j) B_N (x_k)\right) D_N (x_j)\right)
\end{align*}
and 
\[A_1 (x_1) B_1 (x_j) = \lambda(x_1,x_j) B_1 (x_j) A_1 (x_1) + \mu (x_1,x_j) B_1 (x_1) A_1 (x_j).\]

As a consequence of Lemma~\ref{lemma.commutation.relations.monodromy}, 

\begin{align*}
\prod_{k=1}^{n+1} B_{N+1} (x_k) & = 
\left( \prod_{k=1}^{n+1} B_N (x_k) \right) \otimes \left( \prod_{k=1}^{n+1} A_1 (x_k) \right)\\
& \quad + \sum_{j=2}^{n+1}
\lambda (x_1,x_j) \prod_{\underset{k \neq j}{k \ge 2}} \lambda(x_k,x_j). \left( \prod_{\underset{k \neq j}{k \ge 1}}B_N (x_k) D_N (x_j) \right) \otimes \left(B_1 (x_j) \prod_{\underset{k \neq j}{k \ge 1}} A_1 (x_k)\right) \\
& \quad + \sum_{j=2}^{n+1}
\mu (x_1,x_j) \prod_{\underset{k \neq j}{k \ge 2}} \lambda(x_k,x_j). \left( \prod_{\underset{k \neq j}{k \ge 1}}B_N (x_k) D_N (x_j) \right) \otimes \left(B_1 (x_1) \prod_{k \ge 2} A_1 (x_k)\right) 
\\
& \quad + \prod_{k=2}^{n+1} \lambda(x_k,x_1)  \left( \left( \prod_{k=2}^{n+1} B_N(x_k)\right) D_N (x_1) \right)\otimes \left( B_1 (x_1) \prod_{k=2}^{n+1}A_1 (x_k) \right)\\
& \quad + \sum_{j=2}^{n+1} \mu(x_j,x_1). \prod_{\underset{k\neq j}{k \ge 2}} \lambda(x_k,x_j) . \left(\left( \prod_{\underset{k \neq j}{k \ge 1}} B_N (x_k)\right) D_N (x_j) \right) \otimes \left( B_1 (x_1) \prod_{k \ge 2} A_1 (x_k)\right)
\end{align*}

The statement follows by antisymmetry of 
$\mu$.
\end{proof}

\begin{lemma}
\label{lemma.expression.eigenvector}
For all $N$, $0 \le n \le N$, and $\textbf{x}=(x_1,...,x_n) \in 
\mathbb{C}^n$:
\begin{align*}
\psi_{n,N} (\textbf{x}) & =
\left( \prod_{k=1}^n c(x_k) a(x_k)^N \right) \sum_{\boldsymbol{\epsilon} \in \Omega_N^{(n)}}  \left( \sum_{\sigma \in \Sigma_n} \prod_{k=1}^n \left(\frac{b(x_k)}{a(x_k)}\right)^{q_{\sigma(k)}[\boldsymbol{\epsilon}]} \prod_{\sigma(k)<\sigma(j)} \lambda(x_{k},x_{j})\right)\boldsymbol{\epsilon}\\
\end{align*}
\end{lemma}

\begin{proof}

\begin{enumerate}
\item \textbf{A recursion equation 
deriving from recursion on monodromy matrices 
components:}

Let us denote $\epsilon$ the empty sequence,
and set the convention $\psi_{0,0} (\epsilon)$
is neutral for the tensor product (for all vector 
$v$, $\psi_{0,0} (\epsilon) \otimes v = v$).
Since for all $N$ and $1 \le n \le N$ and $\textbf{x}=(x_1,...,x_n)$, 
\[\psi_{n,N}(\textbf{x})= B_N (x_1) ... B_N(x_n) \boldsymbol{\nu}_N,\]
by application of Lemma~\ref{lemma.recursion.monodromy.candidate}, 
we have: 

\begin{align*}
\psi_{n,N+1} (\textbf{x}) & = \left(\prod_{k=1}^n a(x_k) \right). \psi_{n,N} (\textbf{x}) \otimes \ket{0}\\
& \quad + \sum_{j=1}^n b(x_j)^N c(x_j) \left(\prod_{k\neq j} a(x_k) \right)\left( \prod_{k \neq j} \lambda(x_k,x_j) \right) \psi_{n-1,N} (\textbf{x}^{j}_{*}) \otimes \ket{1},\\
\end{align*}

where we applied the equalities $A_1 (x_k).\ket{0} = a(x_k). \ket{0}$, $D_N (x_k) . \boldsymbol{\nu}_N = 
b(x_k)^N . \boldsymbol{\nu}_N$ and $B_1 (x_j). \ket{0} = b(x_j).\ket{0}$.

\item \textbf{Application to a recursion 
on coordinates according to some $\boldsymbol{\epsilon}$:}

It is straightforward to see that 
$\psi_{n,N} (\textbf{x})$ lies in the 
subspace $\Omega_N^{(n)}$, since for 
any $x \in \mathbb{C}$, and $\boldsymbol{\epsilon} 
\in \Omega_N$, we have 
\[\left|B_N (x). \boldsymbol{\epsilon} \right|_1 
\le |\boldsymbol{\epsilon}|_1 +1.\]

As a consequence, in order to prove 
the lemma, it is sufficient to compute 
the coordinate of $\psi_{n,N} (\boldsymbol{x})$ according 
to any $\boldsymbol{\epsilon} \in \Omega_N^{(n)}$.
As a consequence of the first point, if $\boldsymbol{\epsilon}_{N+1} = 0$, then: 

\begin{align*}
\psi_{n,N+1} (\textbf{x}) [\boldsymbol{\epsilon}] & = \left(\prod_{k=1}^n a(x_k) \right). \psi_{n,N} (\textbf{x}) [\boldsymbol{\epsilon}_{*}] \otimes \ket{0}
\end{align*}

Else, if $\boldsymbol{\epsilon}_{N+1} = 1$: 

\begin{align*}
\psi_{n,N+1} (\textbf{x}) [\boldsymbol{\epsilon}] & = \sum_{j=1}^n b(x_j)^N c(x_j) \left(\prod_{k\neq j} a(x_k) \right)\left( \prod_{k \neq j} \lambda(x_k,x_j) \right) \psi_{n-1,N} (\textbf{x}^{j}_{*})  [\boldsymbol{\epsilon}_{*}]\otimes \ket{1},\\
\end{align*}

\item \textbf{Expression of the coefficient of $\psi_{n,N} (\textbf{x})$ relative  to $\boldsymbol{\epsilon}$:}

As a direct consequence of the last point, 
this coefficient is:

\[\sum_{S} \prod_{\boldsymbol{k:\epsilon}_k =1} 
\left( b(x_{l_k[S]})^k c(x_{l_k[S]}) \prod_{j \in S_{k+1}} a(x_j) \prod_{j \in S_{k+1}} \lambda(x_j,x_{l_k[S]})\right)
\prod_{k:\boldsymbol{\epsilon}_k =0} 
\left( \prod_{j \in S_k} a(x_j)\right),\]
where the sum is over the sequences $S=(S_k)_{k \in \llbracket 1,N\rrbracket}$ of subsets of $\{1,...,n\}$
such that $S_N= \emptyset$ and for all $k \le N-1$, $S_{k+1} \subset S_k$ and $|S_{k} \setminus S_{k+1}| = \boldsymbol{\epsilon_k}$. In this 
formula, when $\boldsymbol{\epsilon}_k=1$, we denote $l_k [S]$ the element of $S_{k} \setminus S_{k+1}$. 

\item \textbf{Change of the variable $S$ into 
a permutation of $\Sigma_n$:}

Moreover, a sequence $S=(S_k)_{k \in \llbracket 1,N\rrbracket}$ of subsets of $\{1,...,n\}$ verifies the previous 
hypotheses if and only if 
there exists a permutation 
$\sigma \in \Sigma_n$ such that 
$S_{k+1} = S_k$ whenever 
$\boldsymbol{\epsilon}_k =0$ 
and for all $l \in \{1,...,n\}$, $S_{q_{l} [\boldsymbol{\epsilon}]} = S_{q_{l} [\boldsymbol{\epsilon}]+1} \bigcup \{\sigma(l)\}$.
For this permutation, for all $l \in \{1,...,n\}$,
$S_{q_l [\boldsymbol{\epsilon}]} = \{\sigma(1),...\sigma(l)\}$.

 As a consequence, 
the coefficient is equal to 
\[\sum_{\sigma \in \Sigma_n} \prod_{l=n}^1 
\left( b(x_{\sigma(l)})^{q_l [\boldsymbol{\epsilon}]} c(x_{\sigma(l)}) \prod_{j <l} a(x_{\sigma(j)}) \prod_{j<l} \lambda(x_{\sigma(j)},x_{\sigma(l)})\right)
\prod_{k:\boldsymbol{\epsilon}_k =0} 
\left( \prod_{j \in S_k} a(x_j)\right)\]
This can be rewritten: 
\[\sum_{\sigma \in \Sigma_n} \prod_{l=n}^1 
\left( b(x_{\sigma(l)})^{q_l [\boldsymbol{\epsilon}]} c(x_{\sigma(l)})  \prod_{j<l} \lambda(x_{\sigma(j)},x_{\sigma(l)})\right)
\prod_{k=1}^N 
\left( \prod_{j \in S_{k+1}} a(x_j)\right),\]
where $S_{N+1}=\emptyset$. We rewrite again: 
\[\sum_{\sigma \in \Sigma_n} \prod_{l=n}^1 
\left( b(x_{\sigma(l)})^{q_l [\boldsymbol{\epsilon}]} c(x_{\sigma(l)})  \prod_{j<l} \lambda(x_{\sigma(j)},x_{\sigma(l)})\right)
\prod_{l=1}^n (a(x_{\sigma(l)}))^{N-q_l[\boldsymbol{\epsilon}]},\]
which yields the statement, by the change 
of variable $\sigma \mapsto \sigma^{-1}$.
\end{enumerate}
\end{proof}

\section{\label{section.trigonometric.path} Application to trigonometric Yang-Baxter paths}

In this section, we choose particular 
Yang-Baxter paths for the Lieb path 
$t \mapsto V_N (t)$ defined in Section~\ref{section.lieb.path}, corresponding 
to widely known \textbf{trigonometric 
solutions} of Yang-Baxter equation [Section~\ref{section.path.algebraic.coordinate}]. 
We then prove Theorem~\ref{theorem.coordinate.ansatz}, applying 
the algebraic Bethe ansatz to these Yang-Baxter 
paths [Section~\ref{section.proof}].

\subsection{\label{section.path.algebraic.coordinate}Trigonometric Yang-Baxter path of commuting matrices}

For all $\gamma \in (0,\pi)$ 
and $x \in \mathbb{C}$, consider the local matrix 
function $R_{\gamma}^x$ defined by:

\[R_{\gamma}^x = \frac{1}{\sin(\gamma/2)}\left( \begin{array}{cccc}
\sin(\gamma-x) & 0 & 0 & 0 \\
0 & \sin (x) & \sin(\gamma) & 0 \\
0 & \sin(\gamma) & \sin(x) & 0 \\
0 & 0 & 0 & \sin(\gamma-x)
\end{array}\right)\]

Let us denote $T^x_{\gamma,N}$ 
the corresponding transfer matrix and 
$M^{x}_{\gamma,N}$ the monodromy matrix.
We denote $B_{\gamma,N} (x)$ the north east 
components of the corresponding monodromy 
matrices. Let us also denote 
$\psi_{\gamma,n,N} (x_1,...,x_n)$ the 
candidate eigenvector provided by the algebraic Bethe 
ansatz for the path $x \mapsto T^x_{\gamma,N}$ 
at $\gamma/2$, provided a solution $(x_1,...,x_n)$ 
of the system of Bethe equations, 
and $\Lambda_{\gamma,n,N} (x_1,...,x_n)$ 
the corresponding candidate eigenvalue.
Let us also denote $\delta_0: 
(x,y) \mapsto y-x$.

\begin{lemma}
Let us fix some $\gamma \in  \mathbb{R} \backslash \pi \mathbb{Z}$. 
For all $x,y$, $T^x_{\gamma,N}$ and 
$T^y_{\gamma,N}$ commute.
\end{lemma}

\begin{remark}
These local matrix functions are obtained 
by a parameterization of 
\[\left\{ (a,b,c) \in \mathbb{C}^3: 
a^2 + b^2 - c^2 -2 \Delta ab=0 \right\},\]
where $\Delta= -\cos(\gamma)$.
\end{remark}

\begin{proof}
\begin{itemize}
\item \textbf{Verification of some Yang-Baxter equation:} For all $x,y$, the triple of matrices
$(R_{\gamma}^x,R_{\gamma}^y,
Q_{\gamma}^{\delta_0 (x,y)})$ verifies the Yang-Baxter equation.
Indeed, the system $(S)_{R_{\gamma}^x,R_{\gamma}^y,R_{\gamma}^{\delta_0 (x,y)}}$ is written:

\[\left\{\begin{array}{rcl}
\sin(\gamma-x) \sin(\gamma) \sin(\gamma-(y-x)) & = & \sin(x)\sin(y-x) \sin(\gamma) + \sin(\gamma)^2 \sin(\gamma-y)\\
\sin(\gamma-x) \sin(y) \sin(\gamma) & = & \sin(y) \sin(\gamma-y) \sin(\gamma) + \sin(\gamma)^2 \sin(y-x)\\
\sin(\gamma) \sin(y) \sin(\gamma-x) & = & \sin(\gamma) \sin(\gamma-y) \sin(y-x) + \sin(x) \sin(\gamma)^2
\end{array}\right.\]
By factoring by $\sin(\gamma)$, and developping, 
one can check that this system is verified. 
\item \textbf{Invertibility of 
the north east local matrix:} When $y$ is not in the 
set 
\[\left((x +\gamma) + \pi \mathbb{Z}\right)
\cup \left( \pi-\gamma + 2\pi \mathbb{Z} \right),\]
 $Q_{\gamma}^{y-x}$ 
is invertible. As a consequence, the matrices 
$T^x_{\gamma,N}$ and 
$T^y_{\gamma,N}$ commute.
In the other cases, we obtain this 
commutation by continuity.
\end{itemize}
\end{proof}

For all $t$, we have: 
\[R^{\mu_t /2}_{\mu_t} = \left( \begin{array}{cccc}
1 & 0 & 0 & 0 \\
0 & 1 & t & 0 \\
0 & t & 1 & 0 \\
0 & 0 & 0 & 1
\end{array} \right),\]
which means that the Yang-Baxter path 
of transfer matrices corresponding to 
$x \mapsto R^{x}_{\mu_t}$ intersects 
the Lieb path at $x = \mu_t /2$.
For all $t,x$, we denote 
\[ R^{x}_{\mu_t} \equiv \left( \begin{array}{cccc}
a_t (x) & 0 & 0 & 0 \\
0 & b_t (x) & c_t (x) & 0 \\
0 & c_t (x) & b_t (x) & 0 \\
0 & 0 & 0 & a_t (x)
\end{array}\right).\]

We consider the Yang-Baxter 
path associated to the Lieb path $t \mapsto V_N (t)$ defined as 
\[x \mapsto T^x_N (t) =  T^x_{\mu_t,N}.\]

\subsection{\label{section.proof} Proof of Theorem~\ref{theorem.coordinate.ansatz}}

In this section, we provide a proof of Theorem~\ref{theorem.coordinate.ansatz}, 
which is a consequence of successive applications 
of Theorem~\ref{theorem.identification.equations}, Theorem~\ref{theorem.identification.vectors} and Theorem~\ref{theorem.identification.values}.

\begin{notation}
In this section, we fix some $(p_1,...,p_n) \in I_t^n$ which verifies the system of equations $(E_j)[t,n,N]$, $j \le n$. 
Let us denote, for all $j$, $\alpha_j$ such that 
$\kappa_t (\alpha_j) = p_j$, and 
\[x_j = \frac{\mu_t}{2} + i \frac{\alpha_j}{2}.\] 
\end{notation}

\begin{lemma} 
\label{lemma.expression.a}
For all $j,k$ such that $j \neq k$, we have the following equality:
\[-2 \Delta_t a_t (x_k -x_j) = 2i \sin(\mu_t/2) (b_t (x_j)b_t(x_k) - 2\Delta_t a_t (x_j)b_t (x_k) + a_t (x_j)a_t (x_k)).\]
\end{lemma}

\begin{proof}
\begin{enumerate}
\item \textbf{Expressing terms in the second 
member with exponentials:}

By definition of $a_t$ and $b_t$, for all $j$:

\begin{align*}
2i \sin (\mu_t /2) a_t (x_j) & = e^{i\left(\mu_t - \left(\frac{\mu_t}{2} + i \frac{\alpha_j}{2}\right)\right)} - e^{-i\left(\mu_t - \left(\frac{\mu_t}{2} + i \frac{\alpha_j}{2}\right)\right)}\\
& =  e^{i\frac{\mu_t}{2} + \frac{\alpha_j}{2}} - e^{-i\frac{\mu_t}{2} - \frac{\alpha_j}{2}}\\
2i \sin (\mu_t /2) b_t (x_j) & = e^{i\left(\frac{\mu_t}{2} + i \frac{\alpha_j}{2}\right)} - e^{-i\left(\frac{\mu_t}{2} + i \frac{\alpha_j}{2}\right)} \\
& = e^{i\frac{\mu_t}{2} - \frac{\alpha_j}{2}} - e^{-i\frac{\mu_t}{2} + \frac{\alpha_j}{2}}.
\end{align*}

As a consequence, for all $j,k$ 
such that $j \neq k$:

\begin{align*}
(2i)^2 \sin^2 (\mu_t /2) a_t (x_j) a_t (x_k) & =  e^{-i\mu_t -\frac{\alpha_j}{2} - \frac{\alpha_k}{2}} + 
e^{i \mu_t + \frac{\alpha_j}{2} + \frac{\alpha_k}{2}}  - 
 e^{-\frac{\alpha_k}{2} + \frac{\alpha_j}{2}} - e^{\frac{\alpha_k}{2} - \frac{\alpha_j}{2}}\\
(2i)^2 \sin^2 (\mu_t /2) b_t (x_j) b_t (x_k) & =  e^{i\mu_t -\frac{\alpha_j}{2} - \frac{\alpha_k}{2}} + 
e^{-i \mu_t + \frac{\alpha_j}{2} + \frac{\alpha_k}{2}}  - 
e^{-\frac{\alpha_k}{2} + \frac{\alpha_j}{2}} - e^{\frac{\alpha_k}{2} - \frac{\alpha_j}{2}} \\
(2i)^2 \sin^2 (\mu_t /2) a_t (x_j) b_t (x_k) & = e^{i \mu_t + \frac{\alpha_j}{2} - \frac{\alpha_k}{2}} + e^{-i\mu_t + \frac{\alpha_k}{2} - \frac{\alpha_j}{2}}
- e^{\frac{\alpha_k}{2}+\frac{\alpha_j}{2}} - e^{-\frac{\alpha_k}{2}-\frac{\alpha_j}{2}}
\end{align*}

\item \textbf{Expression of the term in 
$\Delta_t$:}

Thus, since by definition of $\mu_t$ we have$-2\Delta_t = e^{i\mu_t}+e^{-i\mu_t}$:
\[\begin{array}{rcl} -2\Delta_t a_t (x_j) b_t (x_k) (2i)^2  
\sin^2(\mu_t/2) & = &
e^{2i\mu_t + \frac{\alpha_j}{2} - \frac{\alpha_k}{2}} 
+ e^{\frac{\alpha_k}{2}-\frac{\alpha_j}{2}} - e^{i\mu_t +\frac{\alpha_k}{2} +\frac{\alpha_j}{2}} \\
& & - e^{i\mu_t - \frac{\alpha_k}{2} -\frac{\alpha_j}{2}}
+ e^{\frac{\alpha_j}{2}-\frac{\alpha_k}{2}} \\
&  & + e^{-2i\mu_t + \frac{\alpha_k}{2} - \frac{\alpha_j}{2}} - e^{-i \mu_t +\frac{\alpha_k}{2} +\frac{\alpha_j}{2}} - e^{-i \mu_t - \frac{\alpha_k}{2} -\frac{\alpha_j}{2}}
\end{array} 
\]

\item \textbf{Expression of the second member:}

The second member in the equality 
of the statement is thus 
\[\frac{1}{2i \sin (\mu_t/2)} \left( e^{2i\mu_t + \frac{\alpha_j}{2} - \frac{\alpha_k}{2}}  
+ e^{-2i\mu_t + \frac{\alpha_k}{2} - \frac{\alpha_j}{2}} - e^{\frac{\alpha_k}{2} - 
\frac{\alpha_j}{2}} - e^{\frac{\alpha_j}{2} - 
\frac{\alpha_k}{2}}\right) \]

\item \textbf{Expression of the first member:}

On the other hand: 
\begin{align*} 
2i \sin(\mu_t /2) a_t (x_k -x_j) & = 
e^{i \left( \mu_t - i \left( \frac{\alpha_k}{2} - \frac{\alpha_j}{2} \right)\right)} - e^{-i \left( \mu_t - i \left( \frac{\alpha_k}{2} - \frac{\alpha_j}{2}\right)\right)}\\
& = e^{i\mu_t + \frac{\alpha_k}{2} - \frac{\alpha_j}{2}} - e^{-i\mu_t + \frac{\alpha_k}{2} - \frac{\alpha_j}{2}}.
\end{align*}
Thus, using $-2\Delta_t = e^{i\mu_t} - e^{-i\mu_t}$:

\[-2 \Delta_t \cdot 2i \sin(\mu_t /2) a_t (x_k -x_j) =  e^{2i\mu_t + \frac{\alpha_j}{2} - \frac{\alpha_k}{2}}  
+ e^{-2i\mu_t + \frac{\alpha_k}{2} - \frac{\alpha_j}{2}} - e^{\frac{\alpha_k}{2} - 
\frac{\alpha_j}{2}} - e^{\frac{\alpha_j}{2} - 
\frac{\alpha_k}{2}}\]

This yields the statement of the lemma.

\end{enumerate}
\end{proof}

\begin{fact}
\label{fact.quotient}
For all $j \in \{1,...,n\}$, we have, by direct computation, that: 
\[\frac{a_t (x_j)}{b_t (x_j)} = \frac{e^{i\frac{\mu_t}{2}+\frac{\alpha_j}{2}} - e^{-i\frac{\mu_t}{2}-\frac{\alpha_j}{2}}}{e^{i\frac{\mu_t}{2}-\frac{\alpha_j}{2}} - e^{-i\frac{\mu_t}{2}+\frac{\alpha_j}{2}}}= \frac{e^{i\mu_t} - e^{-\alpha_j}}{e^{i\mu_t - \alpha_j} -1} =  e^{-i\kappa_t (\alpha_j)} = e^{-ip_j} \equiv z_j.\]
\end{fact}

\begin{notation}
Let us also denote $x_0 = \frac{\mu_t}{2}$. We have 
directly that:
\[z_0 \equiv \frac{a_t (x_0)}{b_t (x_0)} = 1.\]
\end{notation}

\begin{theorem}
\label{theorem.identification.equations}
Let $(p_1,...,p_n) \in I_t^n$ 
such that $p_1 < ... < p_n$ and for all $j$, 
the equation $(E_j)[t,n,N]$ is verified. 
Then $(x_1,...,x_n)$ verifies the 
equations 
\[(E'_j)[t,n,N]: \qquad 
a_t (x_j)^N \prod_{\underset{k \neq j}{k=1}}^n \frac{a_t (x_k-x_j)}{b_t (x_k-x_j)} = 
b_t (x_j)^N \prod_{\underset{k \neq j}{k=1}}^n \frac{a_t (x_j-x_k)}{b_t (x_j-x_k)}.\]
\end{theorem}

\begin{proof}

\begin{enumerate}
\item \textbf{Rewritting the 
equations $(E'_j)[t,n,N]$:}

By antisymmetry of $(c_t / b_t) \circ \delta_0$, 
and thus $b_t \circ \delta_0$, for all $j,k$ 
such that $j \neq k$:

\[\frac{a_t (x_k-x_j) b_t (x_j-x_k)}{a_t (x_j-x_k) b_t (x_k-x_j)} = - \frac{a_t (x_k-x_j)}{a_t (x_j-x_k)}.\]

By application of Lemma~\ref{lemma.expression.a} 
and Fact~\ref{fact.quotient}:

\begin{align*}
\frac{a_t (x_k-x_j) b_t (x_j-x_k)}{a_t (x_j-x_k) b_t (x_k-x_j)} & = - \frac{ b_t (x_j) b_t (x_k) - 2\Delta_t a_t (x_j) b_t (x_k) + a_t(x_k) a_t (x_j)}{ b_t (x_k) b_t (x_j) - 2\Delta_t a_t (x_k) b_t (x_j) + a_t(x_j) a_t (x_k)}\\
& = - \frac{1-2\Delta_t z_j + z_j z_k}{1-2\Delta_t z_k + z_j z_k}.
\end{align*}

As a consequence, the equation $(E'_j)[t,n,N]$ 
is equivalent to: 

\[z_j ^N = \prod_{k \neq j} \left( - \frac{1-2\Delta_t z_j + z_j z_k}{1-2\Delta_t z_k + z_j z_k}\right).\]

\[z_j ^N = (-1)^{n+1} \prod_{k \neq j} \frac{z_k}{z_j} \left( \frac{1/z_j-2\Delta_t + z_k}{1/z_k-2\Delta_t + z_j}\right).\]

\item \textbf{Rewritting Bethe equations:}

On the other hand, the equation $(E_j)[t,n,N]$ implies, taking the exponential and 
then using the definition of $\Theta_t$:

\[e^{-i N p_j} = (-1)^{n+1} \cdot e^{-i \sum_{k=1}^n \Theta_t (-p_j,-p_k)}.\]

\[e^{-i N p_j} = (-1)^{n+1} \cdot \prod_{k\neq j} 
e^{-i(p_j-p_k)} \frac{e^{-ip_j} + e^{ip_k} - 2\Delta_t}{e^{-ip_k} + e^{ip_j} - 2\Delta_t}.\]

\[z_j ^N  = (-1)^{n+1} \cdot \prod_{k\neq j} 
\frac{z_j}{z_k} \left( \frac{1/z_j-2\Delta_t + z_k}{1/z_k-2\Delta_t + z_j}\right).\]

\end{enumerate}

\end{proof}

\begin{theorem}
\label{theorem.identification.vectors}
Let $(p_1,...,p_n) \in I_t^n$ 
such that $p_1 < ... < p_n$ and for all $j$, 
the equation $(E_j)[t,n,N]$ is verified. 
Then there exists some $\rho \in \mathbb{C}^{*}$ 
such that:
\[\psi_{\mu_t,n,N} (x_1,...,x_n)= \rho \cdot \psi_{t,n,N} (p_1,...,p_n).\]
\end{theorem}

\begin{proof}
By an application of Lemma~\ref{lemma.expression.eigenvector}, 
for all $n$ and $\boldsymbol{\epsilon} \in \Omega_N^{(n)}$:

\[\psi_{\mu_t,n,N} (\textbf{x})[\boldsymbol{\epsilon}] = 
\left( \prod_{k=1}^n c_t(x_k) a_t(x_k)^N \right)
\left( \sum_{\sigma \in \Sigma_n} \prod_{k=1}^n \left( \frac{b_t (x_k)}{a_t (x_k)}\right)^{q_{\sigma(k)}[\boldsymbol{\epsilon}]}  \prod_{\sigma(k)<\sigma(j)} \frac{a_t (x_j-x_k)}{b_t (x_j-x_k)}\right).\]

By the change of variable $\sigma \mapsto \sigma^{-1}$, and that 
for all $k$, 
$b_t (x_{\sigma(k)})/a_t (x_{\sigma(k)}) = 1/z_{\sigma(k)}$:

\begin{align*}\psi_{\mu_t,n,N} (\textbf{x})[\boldsymbol{\epsilon}] & = 
\left( \prod_{k=1}^n c_t(x_k) a_t(x_k)^N \right)
\left( \sum_{\sigma \in \Sigma_n} \prod_{k=1}^n \left( \frac{b_t (x_{\sigma(k)})}{a_t (x_{\sigma(k)})}\right)^{q_k[\boldsymbol{\epsilon}]}  \prod_{k<j} \frac{a_t (x_{\sigma(j)}-x_{\sigma(k)})}{b_t (x_{\sigma(j)}-x_{\sigma(k)})}\right)\\
& = \left( \prod_{k=1}^n c_t(x_k) a_t(x_k)^N \right)
\left( \sum_{\sigma \in \Sigma_n} \prod_{k=1}^n e^{i p_{\sigma(k)} q_k[\boldsymbol{\epsilon}]}  \prod_{k<j} \frac{a_t (x_{\sigma(j)}-x_{\sigma(k)})}{b_t (x_{\sigma(j)}-x_{\sigma(k)})}\right)
\end{align*}

Furthermore, by an application of Lemma~\ref{lemma.expression.a}:
\begin{align*}
\frac{a_t (x_{\sigma(j)}-x_{\sigma(k)})}{b_t (x_{\sigma(j)}-x_{\sigma(k)})} & = -
\frac{b_t (x_{\sigma(j)}) b_t (x_{\sigma(k)}) - 2\Delta_t a_t (x_{\sigma(k)}) b_t (x_{\sigma(j)}) + a_t (x_{\sigma(j)}) a_t (x_{\sigma(k)})}{2\Delta_t \left( e^{(\alpha_{\sigma(k)} - \alpha_{\sigma(j)})/2} - e^{(\alpha_{\sigma(j)} - \alpha_{\sigma(k)})/2}\right)}.\\
& = \frac{-b_t (x_{\sigma(j)}) b_t (x_{\sigma(k)})}{2 \Delta_t \left( e^{(\alpha_{\sigma(k)} - \alpha_{\sigma(j)})/2} - e^{(\alpha_{\sigma(j)} - \alpha_{\sigma(k)})/2}\right)} \cdot \frac{1}{z_{\sigma(k)}}
\left(z_{\sigma(k)} + \frac{1}{z_{\sigma(j)}} - 2\Delta_t \right).\\
\end{align*}

Moreover, 

\[\prod_{k <j} \left( - \frac{b_t (x_{\sigma(j)}) b_t (x_{\sigma(k)})}{2\Delta_t \left( e^{(\alpha_{\sigma(k)} - \alpha_{\sigma(j)})/2} - e^{(\alpha_{\sigma(j)} - \alpha_{\sigma(k)})/2}\right)}\right)\] is equal to 

\[\frac{\prod_{k=1}^n b_t (x_k) ^ {N-1}}{(2\Delta_t)^{(n-1)n/2} (-1)^{|\{(k,j): 
k<j, \sigma(k)>\sigma(j)\}|} \cdot \sqrt{\prod_{k \neq j} \left| e^{(\alpha_k-\alpha_j)/2} 
- e^{(\alpha_j-\alpha_k)/2} \right|^{N-1}}}\]
which is equal to 
\[\frac{\prod_{k=1}^n b_t (x_k) ^ {N-1}}{(2\Delta_t)^{(n-1)n/2} \epsilon(\sigma) \cdot \sqrt{\prod_{k \neq j} \left| e^{(\alpha_k-\alpha_j)/2} 
- e^{(\alpha_j-\alpha_k)/2} \right|^{N-1}}}\]

We deduce, since for all $l$, $z_{\sigma(l)} = 
e^{-p_l}$ that 
\[\psi_{\mu_t,n,N} (\textbf{x})[\boldsymbol{\epsilon}] = \rho \cdot \psi_{t,n,N} (p_1,...p_n),\]
where 
\[\rho = \frac{\prod_{k=1}^n c_t (x_k) a_t (x_k)^N b(x_k)^{N-1}}{(2\Delta_t)^{(n-1)n/2} \sqrt{\prod_{k \neq j} \left| e^{(\alpha_k-\alpha_j)/2} 
- e^{(\alpha_j-\alpha_k)/2} \right|^{N-1}}}.\]
\end{proof}

\begin{theorem}
\label{theorem.identification.values}
Let $(p_1,...,p_n) \in I_t^n$ 
such that $p_1 < ... < p_n$ and for all $j$, 
the equation $(E_j)[t,n,N]$ is verified. 
Then we have the equality:
\[\Lambda_{n,N} (t) [p_1,...,p_n] = \Lambda_{\mu_t,n,N} (x_1,...,x_n).\]
\end{theorem}

\begin{proof}

\begin{itemize}

\item \textbf{When for all $j$, $p_j \neq 0$:}

\begin{enumerate}

\item \textbf{Recall of the definition of 
$\Lambda_{\mu_t,n,N} (x_1,...,x_n)$:} 

\[\Lambda_{\mu_t,n,N} (x_1,...,x_n) 
= a_t (x_0)^N \prod_{k=1}^n \frac{a_t (x_k-x_0)}{b_t (x_k-x_0)} + b_t (x_0)^N \prod_{k=1}^n \frac{a_t (x_0-x_k)}{b_t (x_0-x_k)}.\]

 \item \textbf{Factors in the first product 
 expressed with $L_t$:} for all $j \neq 0$,
  
 \[\frac{a_t (x_j-x_0)}{b_t (x_j-x_0)} = L_t (z_j).\]
 Indeed, for all $j \neq 0$:
 \[\frac{a_t (x_j-x_0)}{b_t (x_j-x_0)} = \frac{a_t (i\frac{\alpha_j}{2})}{b_t (i\frac{\alpha_j}{2})} = \frac{e^{i\mu_t +\frac{\alpha_j}{2}}-e^{-i\mu_t - \frac{\alpha_j}{2}}}{e^{\frac{\alpha_j}{2}}-e^{-\frac{\alpha_j}{2}}}.\]
 
 On the other hand, since $z_i = a_t (x_i)/b_t (x_i)$:
 
 \[L_t (z_i) = 1+ \frac{t^2 z_i}{1-z_i} 
 = 1+ \frac{t^2 a_t (x_i)/b_t (x_i)}{1-a_t (x_i)/b_t (x_i)}
 = \frac{t^2 a_t (x_i) +b_t (x_i) - a_t (x_i)}{b_t (x_i) - a_t (x_i)}\]
 Since $\Delta_t = (2-t^2)/2$, $t^2 = 2-2\Delta_t$:
 \[ 
 L_t(z_i) = \frac{b_t (x_i)+a_t (x_i) -2\Delta_t a_t (x_i)}{b_t (x_i)-a_t (x_i)} \]
 
  \[-2a_t (x_i) \Delta_t = (e^{i\mu_t}+e^{-i\mu_t}).
 (e^{i\frac{\mu_t}{2}+\frac{\alpha_j}{2}} - e^{-i\frac{\mu_t}{2}-\frac{\alpha_j}{2}} )\]
 
 \begin{align*}
 L_t(z_i) &  = \frac{e^{i\frac{\mu_t}{2}-\frac{\alpha_j}{2}} - e^{-i\frac{\mu_t}{2}+\frac{\alpha_j}{2}} 
 +( e^{i\frac{\mu_t}{2}+\frac{\alpha_j}{2}} - e^{-i\frac{\mu_t}{2}-\frac{\alpha_j}{2}})}{e^{i\frac{\mu_t}{2}-\frac{\alpha_j}{2}} - e^{-i\frac{\mu_t}{2}+\frac{\alpha_j}{2}} 
 -( e^{i\frac{\mu_t}{2}+\frac{\alpha_j}{2}} - e^{-i\frac{\mu_t}{2}-\frac{\alpha_j}{2}})} \\
 & \quad   + \frac{
 e^{i3\frac{\mu_t}{2} + \frac{\alpha_j}{2}}  
 + e^{-i \frac{\mu_t}{2} + \frac{\alpha_j}{2}} - e^{i \frac{\mu_t}{2} - \frac{\alpha_j}{2}} - e^{-i 3 \frac{\mu_t}{2} -\frac{\alpha_j}{2}} }{e^{i\frac{\mu_t}{2}-\frac{\alpha_j}{2}} - e^{-i\frac{\mu_t}{2}+\frac{\alpha_j}{2}} 
 -( e^{i\frac{\mu_t}{2}+\frac{\alpha_j}{2}} - e^{-i\frac{\mu_t}{2}-\frac{\alpha_j}{2}})}
 \end{align*}
 
 It is simplified into:

 \begin{align*}
 L_t(z_i)  & = \frac{e^{3i\frac{\mu_t}{2}+\frac{\alpha_j}{2}} - e^{-3i\frac{\mu_t}{2}-\frac{\alpha_j}{2}}
 +( e^{i\frac{\mu_t}{2}+\frac{\alpha_j}{2}} - e^{-i\frac{\mu_t}{2}-\frac{\alpha_j}{2}})}{e^{i\frac{\mu_t}{2}-\frac{\alpha_j}{2}} - e^{-i\frac{\mu_t}{2}+\frac{\alpha_j}{2}} 
 -( e^{i\frac{\mu_t}{2}+\frac{\alpha_j}{2}} - e^{-i\frac{\mu_t}{2}-\frac{\alpha_j}{2}})}\\
 & = \left(\frac{e^{i\frac{\mu_t}{2}} - e^{-i\frac{\mu_t}{2}}}{e^{i\frac{\mu_t}{2}} - e^{-i\frac{\mu_t}{2}}}\right)
\frac{e^{i\mu_t +\frac{\alpha_j}{2}}-e^{-i\mu_t - \frac{\alpha_j}{2}}}{e^{\frac{\alpha_j}{2}}-e^{-\frac{\alpha_j}{2}}}\\
& = \frac{e^{i\mu_t +\frac{\alpha_j}{2}}-e^{-i\mu_t - \frac{\alpha_j}{2}}}{e^{\frac{\alpha_j}{2}}-e^{-\frac{\alpha_j}{2}}}.
 \end{align*}
 
  \item \textbf{Factors in the second product 
 expressed with $M_t$:}

Using similar arguments, we get to:
\[\frac{a_t (x_0-x_j)}{b_t (x_0-x_j)}= M_t (z_i).\]
 
 \item \textbf{Equality to $\Lambda_{n,N} (t) [p_1,...,p_n]$:} 
 
 As a consequence, since $a_t (x_0)=b_t (x_0)=1$:
 
 The candidate eigenvalue is 
 then equal 
 to \[\prod_{k=1}^n L_t (z_k) + \prod_{k=1}^n M_t (z_k) = \Lambda_{n,N} (t) [p_1,...,p_n].\]
 
 \end{enumerate}
 
\item \textbf{When there 
 exists (a unique) $l$ such that $p_l=0$:}
 
 \begin{enumerate}
 
 \item \textbf{Notations:}
 
 Since $p_l=0$, this means, since 
 $\kappa_t$ is increasing and antisymmetric, 
 that $\alpha_l=0$, 
 and thus $x_l = x_0$. In order to prove the theorem, the 
 strategy is to introduce a perturbation to $x_0$. Let us denote, for all $\epsilon>0$, 
 $x_0^{(\epsilon)}
= \frac{\mu_t}{2} + i\frac{\epsilon}{2}$.
Let us also denote 
\begin{align*}
\Lambda^{(\epsilon)} & \equiv \frac{a_t(x_l-x_0^{(\epsilon)})}{b_t(x_l-x_0^{(\epsilon)})}\Lambda^{(\epsilon)}_0+ \frac{a_t(x_0^{(\epsilon)}-x_l)}{b_t(x_0^{(\epsilon)}-x_l)}\Lambda^{(\epsilon)}_1\\
& = \frac{a_t (-i\epsilon/2)}{b_t (-i\epsilon/2)}\Lambda^{(\epsilon)}_0 + \frac{a_t (i\epsilon/2)}{b_t (i\epsilon/2)}\Lambda^{(\epsilon)}_1,
\end{align*}

where 
\[ \Lambda^{(\epsilon)}_0 \equiv a(x_0^{(\epsilon)})^N \prod_{\underset{k \neq l}{k=1}}^n \frac{a_t (x_k -x_0^{(\epsilon)})}{b_t (x_k -x_0^{(\epsilon)})} \qquad 
\text{and} 
\qquad \Lambda^{(\epsilon)}_1 =  b(x_0^{(\epsilon)})^N \prod_{\underset{k \neq l}{k=1}}^n \frac{a_t (x_0^{(\epsilon)}-x_k)}{b_t (x_0^{(\epsilon)}-x_k)}.\]
 
 \item \textbf{Perturbed eigenvalue $\Lambda^{(\epsilon)}$ 
 is in the 
 spectrum of $T_{\mu_t,N} (x_0^{(\epsilon)})$:}
 
By Theorem~\ref{theorem.algebraic.bethe.ansatz},
for all $\epsilon>0$ sufficiently small,

\[T_{\mu_t,N} (x_0^{(\epsilon)}) \cdot 
\psi_{\mu_t,n,N} (\textbf{x}) = \Lambda^{(\epsilon)} \cdot \psi_{\mu_t,n,N} (\textbf{x})\]

It is natural to expect that $\Lambda^{(\epsilon)}$ 
admits a limit when $\epsilon$ tends to zero. 
However, this computation is slightly more 
complex than in the first case (when all the $p_j$ 
are non zero).

\item \textbf{Developping the expression of 
$\Lambda^{(\epsilon)}$:}

We have the following:
\begin{align*}
\frac{a_t (-i\epsilon/2)}{b_t (-i\epsilon/2)} & = \frac{e^{i\mu_t  - \epsilon/2} - e^{-i\mu_t + \epsilon/2}}{e^{\epsilon/2}-e^{-\epsilon/2}}
= - e^{i\mu_t} +
\frac{e^{i\mu_t+\epsilon/2}-e^{-i\mu_t+\epsilon/2}}{e^{\epsilon/2}-e^{-\epsilon/2}}\\
& = - e^{i\mu_t} + e^{\epsilon}
 \frac{e^{i\mu_t}-e^{-i\mu_t}}{e^{\epsilon}-1}
\end{align*}

Similarly:

\[\frac{a_t (i\epsilon/2)}{b_t (i\epsilon/2)} =  
-e^{i\mu_t} -\frac{e^{i\mu_t}-e^{-i\mu_t}}{e^{\epsilon}-1}.\]

As a consequence: 

\begin{align*}
\Lambda^{(\epsilon)} &  = - e^{i\mu_t} \left( \Lambda^{(\epsilon)}_0 + \Lambda^{(\epsilon)}_1 \right)  
+ e^{\epsilon} \frac{e^{i\mu_t}-e^{-i\mu_t}}{e^{\epsilon}-1} \Lambda^{(\epsilon)}_0 - \frac{e^{i\mu_t}-e^{-i\mu_t}}{e^{\epsilon}-1} \Lambda^{(\epsilon)}_1 \\
& = - e^{i\mu_t} \left( \Lambda^{(\epsilon)}_0 + \Lambda^{(\epsilon)}_1 \right) +  e^{\epsilon} \frac{e^{i\mu_t}-e^{-i\mu_t}}{e^{\epsilon}-1}\left( 
\Lambda^{(\epsilon)}_0 - \Lambda^{(\epsilon)}_1 \right) + (e^{i\mu_t}-e^{-i\mu_t}) \Lambda^{(\epsilon)}_1\\
& = -e^{i\mu_t} \Lambda_0^{(\epsilon)} - e^{-i\mu_t} \Lambda_1^{(\epsilon)} + e^{\epsilon} \frac{e^{i\mu_t}-e^{-i\mu_t}}{e^{\epsilon}-1}\left( 
\Lambda^{(\epsilon)}_0 - \Lambda^{(\epsilon)}_1 \right)\\
& = - e^{i\mu_t} \Lambda_0^{(\epsilon)} - e^{-i\mu_t} \Lambda_1^{(\epsilon)} + e^{\epsilon} \frac{e^{i\mu_t}-e^{-i\mu_t}}{e^{\epsilon}-1} \Lambda^{(\epsilon)}_1 \left( 
\frac{\Lambda^{(\epsilon)}_0}{\Lambda^{(\epsilon)}_1} - 1 \right).
\end{align*}

The only term whose limit is unknown is the third 
one, on which we will focus.

\item \textbf{An expression of $\Lambda^{(\epsilon)}_0/ \Lambda^{(\epsilon)}_1$ 
using the function $\Theta_t$:}

Let us denote $\chi$ the function  
\[\chi : \alpha \mapsto (n+1)\pi + N \kappa_t (\alpha) 
+ \sum_{\underset{k\neq l}{k=1}}^n \Theta_t (\kappa_t (\alpha),p_k).\]

By definition of $\Theta_t$, antisymmetry of $\kappa_t$ and Lemma~\ref{lemma.expression.a}:

\begin{align*}
\frac{\Lambda^{(\epsilon)}_0}{\Lambda^{(\epsilon)}_1} 
& = \frac{a_t (x_0^{(\epsilon)})^N}{b_t (x_0^{(\epsilon)})^N} \prod_{\underset{k\neq l}{k=1}}^n \frac{a_t (x_k-x_0^{(\epsilon)})}{b_t (x_k-x_0^{(\epsilon)})} \frac{b_t (x_0^{(\epsilon)}-x_k)}{a_t (x_0^{(\epsilon)}-x_k)}\\
& = \left(\frac{ e^{i\frac{\mu_t}{2}+\frac{\epsilon}{2}}-e^{-i\frac{\mu_t}{2}-\frac{\epsilon}{2}}}{e^{i\frac{\mu_t}{2}-\frac{\epsilon}{2}}-e^{-i\frac{\mu_t}{2}+\frac{\epsilon}{2}}}\right)^N \prod_{\underset{k\neq l}{k=1}}^n \frac{a_t (x_k-x_0^{(\epsilon)})}{b_t (x_k-x_0^{(\epsilon)})} \frac{b_t (x_0^{(\epsilon)}-x_k)}{a_t (x_0^{(\epsilon)}-x_k)}\\
& = (-1)^{n+1} \left(\frac{ e^{i \mu_t}-e^{-\epsilon}}{e^{i \mu_t -\epsilon}-1}\right)^N 
\exp(i\left( \sum_{\underset{k\neq l}{k=1}}^n \Theta_t (\kappa_t (\epsilon),p_k)\right))\\
& = \exp(i \left( -(n+1)\pi + N \kappa_t (-\epsilon) - \sum_{\underset{k\neq l}{k=1}}^n \Theta_t (\kappa_t (\epsilon),p_k) \right))\\
& = \exp(i \left( -(n+1)\pi - N \kappa_t (\epsilon) - \sum_{\underset{k\neq l}{k=1}}^n \Theta_t (\kappa_t (\epsilon),p_k) \right))
\end{align*}

We deduce that $\Lambda^{(\epsilon)}_0/ \Lambda^{(\epsilon)}_1 = e^{-i\chi(\epsilon)}$.

When $\epsilon=0$, by Bethe equations, $\chi(\epsilon)=0$. 

\item \textbf{First term of Taylor expansion:}

We deduce that $\chi(\epsilon) = \chi'(0) \epsilon + o(\epsilon)$. As a consequence, 
\[ \frac{\Lambda^{(\epsilon)}_0}{\Lambda^{(\epsilon)}_1} = 
1 - i\chi'(0) \epsilon + o(\epsilon).\]

Since when $\epsilon \rightarrow 0$, $\Lambda_0^{(\epsilon)}$ converges towards 
$\prod_{k \neq l} M(z_k)$, which is equal by Bethe equations
to $\prod_{k \neq l} L(z_k)$ - the limit of $\Lambda_1^{(\epsilon)}$ - $\Lambda^{(\epsilon)}$ 
tends towards 
\begin{align*}
\Lambda_{n,N} (t) [p_1,...,p_n] & = - (e^{i\mu_t}+ e^{-i\mu_t}) \cdot \prod_{k \neq l} M(z_k) - i (e^{i\mu_t}-e^{-i\mu_t}) \cdot \chi'(0) \cdot \prod_{k \neq l} M(z_k) \\
\end{align*}

Since $\kappa'_t (0) = \frac{\sin(\mu_t)}{1-\cos(\mu_t)}$, 
and using trigonometric equalities the equalities \[\Delta_t = -\cos(\mu_t) = (2-t^2)/2,\]
we have the following:

\begin{align*}
\Lambda_{n,N} (t) [p_1,...,p_n] & = \left(  - 2\cos(\mu_t) + 2 \frac{\sin^2 (\mu_t)}{1-\cos(\mu_t)} \cdot \left( N  + \sum_{k \neq l} \frac{\partial \Theta_t}{\partial x} (0,p_k)\right)\right) \cdot \prod_{k \neq l} M(z_k)\\
& = \left( - 2\cos(\mu_t) + 2 (1+ \cos(\mu_t)) \cdot  \left( N  + \sum_{k \neq l} \frac{\partial \Theta_t}{\partial x} (0,p_k)\right)\right) \cdot \prod_{k \neq l} M(z_k)\\
& = \left( (2-t^2) + t^2 \cdot \left(  N  + \sum_{k \neq l} \frac{\partial \Theta_t}{\partial x} (0,p_k) \right) \right) \cdot \prod_{k \neq l} M(z_k)\\ 
& = \left( 2 + t^2 (N-1) + t^2 \sum_{k \neq l} \frac{\partial \Theta_t}{\partial x} (0,p_k) \right) \cdot \prod_{k \neq l} M(z_k)\\
\end{align*}

\end{enumerate}

 \end{itemize}
\end{proof}

\section{\label{section.commutation} Commutation of the transfer matrix with some hamiltonian}

The aim of this section is to prove that for all $t$, the matrix 
$V_N (t)$ commutes with a Heisenberg hamiltonian $H_N (t)$ [Section~\ref{section.proof.commutation}], which 
is defined in Section~\ref{section.definitions.hamiltonian}.
The proof relies on the paths of commuting matrices $x \mapsto T_{\mu_t,N}^x$.
We also provide an expression of the eigenvalue of $H_N (t)$ 
associated to the candidate eigenvector obtained by the algebraic Bethe ansatz.

\subsection{\label{section.definitions.hamiltonian} Definitions}

Let us recall that $\Omega_N = \mathbb{C}^2 \otimes ... 
\otimes \mathbb{C}^2$. In this section, for the purpose of 
notation, we identify $\{1,...,N\}$ with 
$\mathbb{Z}/N\mathbb{Z}$.

\begin{notation}
Let us denote $a$ and $a^{*}$ the matrices in $\mathcal{M}_{2} (\mathbb{C})$ 
equal to 
\[ a \equiv \left(\begin{array}{cc} 0 & 0 \\ 1 & 0 \end{array}\right), \quad  a^{*} \equiv \left(\begin{array}{cc} 0 & 1 \\ 0 & 0 \end{array}\right).\]
For all $j \in \mathbb{Z}/N\mathbb{Z}$, we denote 
$a_j$ (\textbf{creation} operator at position $j$) and $a^{*}_j$ (\textbf{anihilation} operator at position $j$) the matrices in $\mathcal{M}_{2^N}(\mathbb{C})$ equal 
to 
\[a_j \equiv id \otimes ... \otimes a \otimes ... \otimes id, \quad a^{*}_j \equiv id \otimes ... \otimes a^{*} \otimes ... \otimes id.\]
where $id$ denotes the identity, and $a$ acts on 
the $j$th copy of $\mathbb{C}^2$.
\end{notation}

In other words, the image 
of a vector $\ket{\boldsymbol{\epsilon}_1 ... \boldsymbol{\epsilon}_N}$
in the basis of $\Omega_N$ by $a_j$ (resp. $a^{*}_j$) is as follows:

\begin{itemize}
\item if $\boldsymbol{\epsilon}_j=0$ (resp. $\boldsymbol{\epsilon}_j = 1$), then the image vector is $\textbf{0}$; 
\item if $\boldsymbol{\epsilon}_j=1$ (resp. $\boldsymbol{\epsilon}_j = 0$), then the image vector is $\ket{\boldsymbol{\eta}_1 ... \boldsymbol{\eta}_N}$ such that 
$\boldsymbol{\eta}_j = 0$ (resp. $\boldsymbol{\eta}_j = 1$) and 
for all $k \neq j$, $\boldsymbol{\eta}_k = \boldsymbol{\epsilon}_k$. 
\end{itemize}

\begin{remark}
The term creation (resp. anihilation) refer to 
the fact that for two elements $\boldsymbol{\epsilon}$, $\boldsymbol{\eta}$ of the 
basis of $\Omega_N$, $a_j [\boldsymbol{\epsilon},\boldsymbol{\eta}] \neq 0$ (resp. $a_j^{*} [\boldsymbol{\epsilon},\boldsymbol{\eta}] \neq 0$) implies that 
$|\boldsymbol{\eta}|_1 = |\boldsymbol{\eta}|+1$
(resp. $|\boldsymbol{\eta}|_1 = |\boldsymbol{\eta}|-1$). If we think of $1$ symbols 
as particles, this operator acts by creating (resp. anihilating)
a particle.
\end{remark}

\begin{definition}
\label{definition.hamiltonian}
Let us denote $H_N$ the matrix in $\mathcal{M}_{2^N} (\mathbb{C})$ defined as:
\[H_N (t) = \displaystyle{\sum_{j \in 
\mathbb{Z}/N\mathbb{Z}} \left(a^{*}_{j} a_{j+1} + a_{j} a^{*}_{j+1}\right) + D_t},\]
where $D_t$ is the diagonal matrix such that for all $\boldsymbol{\epsilon}$ in the canonical basis of $\Omega_N$, 
\[D_t [\boldsymbol{\epsilon},\boldsymbol{\epsilon}] = \Delta_t \cdot \left|\{i \in \mathbb{Z}/N \mathbb{Z} : \boldsymbol{\epsilon}_i = \boldsymbol{\epsilon}_{i+1}\}\right|.\]
\end{definition}

\begin{remark}
Let us notice that the hamiltonian defined in Definition~\ref{definition.hamiltonian} differs by a multiple of the 
identity matrix from the hamiltonian defined in~\cite{Duminil-Copin.ansatz} (equation (2.9)). Hence, the commutation property 
which is proved in the present text is equivalent to the 
commutation of $V_N (t)$ with the Hamiltonian defined in~\cite{Duminil-Copin.ansatz}.
\end{remark}

\subsection{\label{section.proof.commutation} Proof of the commutation}

In this section, we fix an integer $N \ge 2$ and a positive number $t$.

\begin{notation} 
In the following of this section, we 
denote $R : x \mapsto R_{\mu_t}^{x}$ and
$T : x \mapsto T_{\mu_t,N}^{x}$.
\end{notation}

\begin{proposition}
\label{proposition.commutation}
We have the equality
\[H_N(t) = t \sin(\mu_t/2) T'(0) T(0)^{-1}.\]
Moreover, $H_N(t)$ and $V_N (t)$ 
commute: 
\[H_N (t) \cdot V_N (t) = V_N (t) \cdot H_N (t).\]
\end{proposition}

\noindent \textbf{Idea of the proof:} \textit{Another 
proof of this proposition can be found in~\cite{Duminil-Copin.ansatz} (Lemma 5.1). The proof that we propose here uses the commutation 
of the transfer matrices of the Yang-Baxter path $x \mapsto 
T^{x}_{\mu_t,N}$. From this commutation property, 
we obtain the commutation of the transfer matrix $V_N (t)$ 
with $T'(0)$ and $T(0)$. Since $H_N (t)$ is expressed with 
$T'(0)$ and $T(0)$, $H_N (t)$ commutes with $V_N (t)$. 
This expression of $H_N (t)$ derives from the analysis of the 
action of the matrices $T(0)$ 
and $T'(0)$.}

\begin{proof}
\begin{enumerate}
\item \textbf{The matrix $T(0)$:}

This transfer matrix 
is constructed from the local matrix function whose representation 
is: 
\[R^0_{\mu_{t}} = \frac{\sin(\mu_{t})}{\sin(\mu_{t}/2)}\left( \begin{array}{cccc}  1  & 0 & 0 & 0 \\
0 & 0 & 1 & 0 \\
0 & 1 & 0 & 0 \\
0 & 0 & 0 & 1
\end{array}\right) = t \cdot P,\]
where we denote 
\[P \equiv \left( \begin{array}{cccc}  1  & 0 & 0 & 0 \\
0 & 0 & 1 & 0 \\
0 & 1 & 0 & 0 \\
0 & 0 & 0 & 1
\end{array}\right).\]
Equivalently, this matrix has entry $T(0)[\boldsymbol{\epsilon},\boldsymbol{\eta}] = t^N$ when there is a $(N,1)$-cylindric pattern 
of $X^s$ that connects $\boldsymbol{\epsilon}$ to $\boldsymbol{\eta}$ which does 
not contain the symbols \begin{tikzpicture}[scale=0.4,baseline=1mm] 
\draw[line width =0.5mm,color=gray!80] (0,0.5) -- (1,0.5);
\draw (0,0) rectangle (1,1);
\end{tikzpicture} and \begin{tikzpicture}[scale=0.4,baseline=1mm] 
\draw[line width =0.5mm,color=gray!80] (0.5,0) -- (0.5,1);
\draw (0,0) rectangle (1,1);
\end{tikzpicture} ; in other words, $\boldsymbol{\eta}$ is 
obtained from $\boldsymbol{\epsilon}$ by shifting all its symbols $1$  
by one position to the left. Otherwise, $T(0)[\boldsymbol{\epsilon},\boldsymbol{\eta}]=0$.

\item \textbf{The matrix $T'(0)$:}

The derivative of 
$R: x \mapsto R^x_{\mu_{t}}$ is the function that to $x$ associates:

\[R'(x) = \frac{1}{\sin(\mu_t/2)} \left( \begin{array}{cccc}  -\cos\left( \mu_t -x\right)  & 0 & 0 & 0 \\
0 & \cos(x) & 0 & 0 \\
0 & 0 & \cos(x) & 0 \\
0 & 0 & 0 & -\cos\left( \mu_t -x\right)
\end{array}\right).\] 

In particular it has the following value in $0$: 

\[R'(0) = \frac{1}{\sin(\mu_t/2)}\left( \begin{array}{cccc}  \Delta_t  & 0 & 0 & 0 \\
0 & 1 & 0 & 0 \\
0 & 0 & 1 & 0 \\
0 & 0 & 0 & \Delta_t
\end{array}\right) \equiv \left( \begin{array}{cc} R'(0)(0,0) & R'(0) (0,1)\\
R'(0) (1,0) & R'(0)(1,1)
\end{array}\right).\]

Let us recall that for all $\boldsymbol{\epsilon},\boldsymbol{\eta}$,
by definition of the transfer matrix from the monodromy matrix:
\[T (x) [\boldsymbol{\epsilon},\boldsymbol{\eta}] = \prod_{k=0}^{N-1} R^x_{\mu_{t}}(\boldsymbol{\epsilon}_k,\boldsymbol{\eta}_k) [0,0]+\prod_{k=0}^{N-1} R^x_{\mu_{t}}(\boldsymbol{\epsilon}_k,\boldsymbol{\eta}_k) [1,1].\]

As a consequence, 
\begin{align*}
T'(0)[\boldsymbol{\epsilon},\boldsymbol{\eta}] &  = \sum_{j=0}^{N-1}  \left(\prod_{k=0}^{j-1} R^0_{\mu_{t}}(\boldsymbol{\epsilon}_k,\boldsymbol{\eta}_k) \right) R'(0) (\boldsymbol{\epsilon}_j,\boldsymbol{\eta}_j) \left(\prod_{k=j+1}^{N-1} R^0_{\mu_{t}}(\boldsymbol{\epsilon}_k,\boldsymbol{\eta}_k)\right) [0,0]\\
& \quad + \left(\prod_{k=0}^{j-1} R^0_{\mu_{t}}(\boldsymbol{\epsilon}_k,\boldsymbol{\eta}_k) \right) R'(0) (\boldsymbol{\epsilon}_j,\boldsymbol{\eta}_j) \left(\prod_{k=j+1}^{N-1} R^0_{\mu_{t}}(\boldsymbol{\epsilon}_k,\boldsymbol{\eta}_k)\right) [1,1].
\end{align*}

With a similar interpretation as for $T(0)$, 
the $j$th term of this sum can be different from $0$ only if 
$\boldsymbol{\eta}$ is obtained from $\boldsymbol{\epsilon}$ by shifting 
any of its $1$ symbols by one position to the left, except 
potentially for the $j-1$th and $j$th positions. Regarding 
these positions, the following cases are possible:  

\begin{enumerate}
\item if 
$\boldsymbol{\epsilon}_{j-1} = 1$ and $\boldsymbol{\epsilon}_{j} = 0$:  
$\boldsymbol{\eta}$ has to be obtained from $\boldsymbol{\epsilon}$ 
by shifting all the curves by one position to the left, except 
the one on position $j-1$, which is shifted to position $j+1$.
In this case, the $j$ term in the sum is $t^{N-1} \cdot \frac{1}{\sin(\mu_t/2)}$.
\item if 
$\boldsymbol{\epsilon}_{j-1} = 0$ and $\boldsymbol{\epsilon}_{j} = 0$: 
$\boldsymbol{\eta}$ is obtained from $\boldsymbol{\epsilon}$ by shifting 
any of its $1$ symbols by one position to the left, and 
the $j$th term is equal to $t^{N-1} \frac{\Delta_t}{\sin(\mu_t/2)}$.
\item if 
$\boldsymbol{\epsilon}_{j-1} = 0$ and $\boldsymbol{\epsilon}_{j} = 1$: 
$\boldsymbol{\eta}$ is obtained from $\boldsymbol{\epsilon}$ by shifting 
any of its $1$ symbols by one position to the left, except the one 
on position $j$, which is fixed, and 
the $j$th term of the sum 
is equal to $t^{N-1} \frac{1}{\sin(\mu_t/2)}$.
\item if 
$\boldsymbol{\epsilon}_{j-1} = 1$ and $\boldsymbol{\epsilon}_{j} = 1$: 
$\boldsymbol{\eta}$ is obtained from $\boldsymbol{\epsilon}$ by shifting 
any of its $1$ symbols by one position to the left, and 
the $j$th term is equal to $t^{N-1} \frac{\Delta_t}{\sin(\mu_t/2)}$.
\end{enumerate} 

\item \textbf{Expression of $H_N (t)$ with ${T(0)}^{-1}$ and $T'(0)$:}

From the definition of $T(0)$, this matrix is invertible. 
Moreover, as a consequence of the last point, we have: 
\[H_N (t) = t\sin(\mu_t/2) \cdot T(0)^{-1} \cdot T'(0).\] 

\item \textbf{Commutation of $H_N (t)$ with $V_N (t)$:}

Since for all $x$, $T(x)$ commutes with $V_N (t)$ (by construction of the Yang-Baxter path):
\[V_N (t) . \frac{T(x)-T(0)}{x} 
=  \frac{T(x)-T(0)}{x} . V_N (t).\]
As a consequence, taking the limit $x \rightarrow 0$, $V_N (t)$ commutes with $T'(0)$, thus with $T(0)^{-1} \cdot T'(0)$, and 
thus with $H_N (t)$.
\end{enumerate}
\end{proof}

\begin{proposition}
Let $n\le N$ and $(p_1,...p_n) \in I_t^n$ such that 
$p_1 < ... < p_n$ and for all $j$, the 
equation $(E_j) [t,n,N]$ is verified. Then:

\[H_N (t). \psi_{t,n,N} (p_1,...,p_n) = t \sin(\mu_t/2) \cdot \left( N \frac{\cos(\mu_t)}{\sin(\mu_t)} + 2 \frac{1}{\sin(\mu_t)^{N+1}} \sum_{j=1}^n (\cos(\mu_t)+\cos(p_j))\right). \psi_{t,n,N} (p_1,...,p_n) \]
\end{proposition}

\begin{remark}
In particular, when $t= \sqrt{2}$, since $\mu_t=\pi/2$ we have: 
\[H_N (\sqrt{2}). \psi_{\sqrt{2},n,N} (p_1,...,p_n) =  \left(2 \sum_{k=1}^n \cos(p_k)\right). \psi_{\sqrt{2},n,N} (p_1,...,p_n),\]
used in the proof of the value of square ice entropy~\cite{Gangloff2019}.
\end{remark}

\begin{proof}
In this proof, we denote $\psi \equiv \psi_{t,n,N} (p_1,...,p_n)$
\begin{enumerate}
\item \textbf{Eigenvalue of $\psi$ for $T(0)$:}

By applying Theorem~\ref{theorem.algebraic.bethe.ansatz} to 
the trigonometric local matrix function $R_{\mu_t}^0$ and $\textbf{x} = (x_1,...,x_n)$ 
such that for all $j$, 
$x_j = \frac{\mu_t}{2} + i \frac{\alpha_j}{2}$, 
with $\alpha_j=\kappa_t ^{-1} (p_j)$, 
for all $x$, 

\[T(ix/2) \cdot \psi = \left( a_t (ix/2)^N \prod_{k=1}^n 
\frac{a_t (x_k-ix/2)}{b_t(x_k-ix/2)}
+ b_t (ix/2)^N \prod_{k=1}^n 
\frac{a_t (ix/2-x_k)}{b_t(ix/2-x_k)}\right) \cdot \psi.\]

As a consequence, we have: 
\begin{align*}T (0). \psi_{t,n,N} (p_1,...,p_n) & = \sin(\mu_t)^N \cdot \left(\prod_{k=1}^{n} \frac{a_t (x_k)}{b_t (x_k)}\right) \cdot \psi_{t,n,N} (p_1,...,p_n)\\
& = \sin(\mu_t)^N \cdot \left(\prod_{k=1}^{n} e^{-ip_k}\right) \cdot \psi_{t,n,N} (p_1,...,p_n)
\end{align*}

\item \textbf{Eigenvalue of $\psi$ for $T'(0)$:}

For all $k$ and $x \neq 0$ sufficiently close to $0$, 
\[\frac{a_t (x_k-ix/2)}{b_t (x_k-ix/2)} = \frac{e^{i\frac{\mu_t}{2} +\frac{\alpha_j}{2} + \frac{x}{2}} - e^{-i\frac{\mu_t}{2} -\frac{\alpha_j}{2} - \frac{x}{2}}}{e^{i\frac{\mu_t}{2} -\frac{\alpha_j}{2} + \frac{x}{2}} - e^{-i\frac{\mu_t}{2} +\frac{\alpha_j}{2} - \frac{x}{2}}} = e^{i\kappa_t (x-\alpha_j)}\] 
Deriving the equality of the first equality in the 
first point relatively to $x$ and evaluating in $0$ (let 
us notice that since $N \ge 2$, the derivative 
of the second sum in $0$ is equal to zero), we obtain, using 
the antisymmetry of $\kappa_t$ and thus the symmetry of $\kappa'_t$:
\[\frac{i}{2} T'(0) \cdot \psi =  \left(\frac{i}{2} N \sin(\mu_t)^{N-1} 
\cos(\mu_t) \prod_{k=1}^n e^{-ip_k} + i\sum_{j=1}^n \kappa'_t (\alpha_j) \prod_{k=1}^n e^{-ip_k} \right) \cdot \psi.\]

\[ T'(0) \cdot \psi =  \left(N \sin(\mu_t)^{N-1} 
\cos(\mu_t) + 2 \sum_{j=1}^n \kappa'_t (\alpha_j) \right) \prod_{k=1}^n e^{-ip_k} \cdot \psi.\]

\textbf{Eigenvalue of $\psi$ for $H_N (t)$:}

Since for all $j$, $\kappa'_t (\kappa^{-1}_t (p_j)) = (\cos(p_j)+\cos(\mu_t))/\sin(\mu_t)$, 
and by an application of Proposition~\ref{proposition.commutation}: 

\[H_N (t).\psi = t \sin(\mu_t/2) \cdot \left( N \frac{\cos(\mu_t)}{\sin(\mu_t)} + 2 \frac{1}{\sin(\mu_t)^{N+1}} \sum_{j=1}^n (\cos(\mu_t)+\cos(p_j))\right) \cdot \psi.\]
\end{enumerate}
\end{proof}

\end{document}